%% file: paper.tex
\documentclass[12pt]{article}
\usepackage[dvips]{graphicx}
\usepackage{amsmath,amsfonts,amssymb,latexsym,epsfig}
\include{references}
\usepackage{kpfonts}

\usepackage{fullpage}
\usepackage{amsthm}    % need for subequations
\usepackage{mathabx}    % need for subequations
\usepackage{mathtools}
\usepackage{graphicx}   % need for figures
\usepackage{verbatim}   % useful for program listings
\usepackage{color}      % use if color is used in text
\usepackage{caption,subcaption}

\theoremstyle{definition}

\newtheorem{theorem}{Theorem}[section]
\newtheorem{lemma}[theorem]{Lemma}
\newtheorem{example}[theorem]{Example}
\newtheorem{remark}[theorem]{Remark}

\newtheorem{proposition}[theorem]{Proposition}
\newtheorem{assumption}[theorem]{Assumption}

\numberwithin{equation}{section}
\usepackage{resizegather}

\begin{document}
\title{Brownian Motion in an N-scale periodic Potential}
\author{A. B. Duncan, G.A. Pavliotis \\
        Department of Mathematics\\
    Imperial College London \\
        London SW7 2AZ, UK 
                           }
\maketitle

% REQUIRED
\begin{abstract}
We study the problem of Brownian motion in a multiscale potential. The potential is assumed to have $N+1$ scales (i.e. $N$ small scales and one macroscale) and to depend periodically on all the small scales. We show that for nonseparable potentials, i.e. potentials in which the microscales and the macroscale are fully coupled,  the homogenized equation is an overdamped Langevin equation with multiplicative noise driven by the free energy, for which the detailed balance condition still holds.  The calculation of the effective diffusion tensor requires the solution of a system of $N$ coupled Poisson equations.
\end{abstract}

% REQUIRED
% \begin{keywords}
%   Brownian dynamics, multiscale analysis, reiterated homogenization, reversible diffusions, free energy.
% \end{keywords}

% % REQUIRED
% \begin{AMS}
% 35B27,35Q82,60H30
% \end{AMS}

\section{Introduction}
The evolution of complex systems arising in chemistry and biology often involve dynamic phenomena occuring at a wide range of time and length scales. Many such systems are characterised by the presence of a hierarchy of barriers in the underlying energy landscape, giving rise to a complex network of metastable regions in configuration space.   Such energy landscapes occur naturally in macromolecular models of solvated systems, in particular protein dynamics. In such cases the rugged energy landscape is due to the many competing interactions in the energy function \cite{bryngelson1995funnels}, giving rise to frustration, in a manner analogous to spin glass models \cite{bryngelson1987spin,onuchic1997theory}.  Although the large scale structure will determine the minimum energy configurations of the system, the small scale fluctuations of the energy landscape will still have a significant influence on the dynamics of the protein, in particular the behaviour at equilibrium,  the most likely pathways for binding and folding, as well as the stability of the conformational states.   Rugged energy landscapes arise in various other contexts, for example nucleation at a phase transition and solid transport in condensed matter.
\\\\
To study the influence of small scale potential energy fluctuations on the system dynamics, a number of simple mathematical models have been proposed which capture the essential features of such systems.   In one such model, originally proposed by Zwanzig \cite{zwanzig1988diffusion},  the dynamics are modelled as an overdamped Langevin diffusion in a rugged two--scale potential $V^\epsilon$, 
\begin{equation}
\label{eq:sde1}
  dX^\epsilon_t = -\nabla V^\epsilon(X_t)\,dt + \sqrt{2\sigma}\,dW_t,\quad \sigma = \beta^{-1} = k_BT,
\end{equation}
where $T$ is the temperature and $k_B$ is Boltmann's constant.  The function $V^\epsilon(x) = V(x,x/\epsilon)$ is a smooth potential which has been perturbed by a rapidly fluctuating function with wave number controlled by the small scale parameter $\epsilon > 0$.  See Figure \ref{fig:potential} for an illustration.  Zwanzig's analysis was based on an effective medium approximation of the mean first passage time, from which the standard Lifson-Jackson formula \cite{lifson1962self} for the effective diffusion coefficient was recovered.  In the context of protein dynamics,  phenomenological models based on (\ref{eq:sde1}) are widespread in the literature, including but not limited to \cite{ansari2000mean,hyeon2003can,mondal2009noise,saven1994kinetics}.  Theoretical aspects of such models have also been previously studied.   More recent studies include \cite{dean2014diffusion} where the authors study diffusion in a strongly correlated quenched random potential constructed from a periodically-extended path of a fractional Brownian motion, and \cite{banerjee2014diffusion} in which the authors perform a numerical study of the effective diffusivity of diffusion in a potential obtained from a realisation of a stationary isotropic Gaussian random field.
\begin{figure}[h!]
\includegraphics[width=1.0\textwidth]{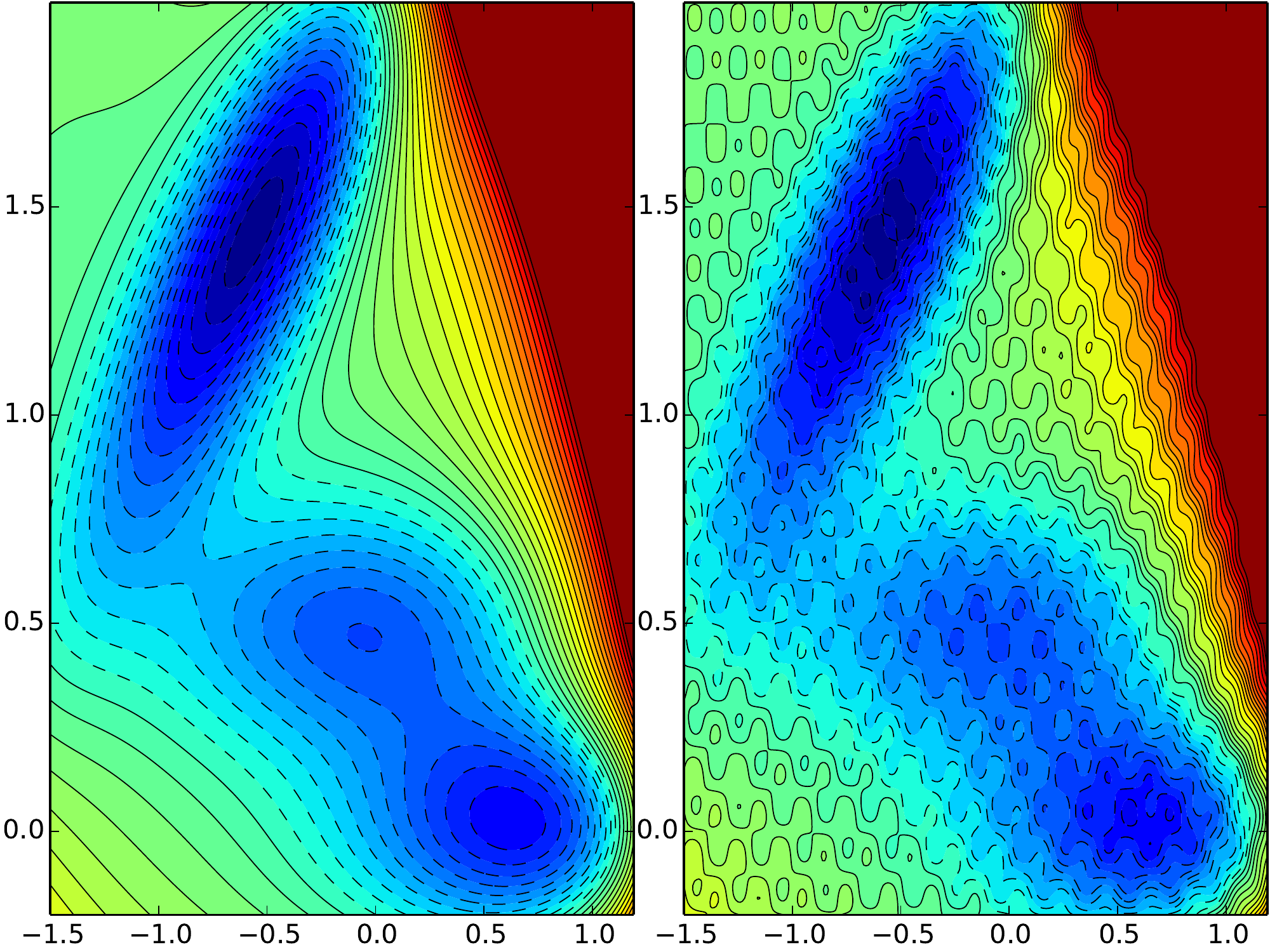}
\centering
\caption{Example of a multiscale potential.  The left panel shows the isolines of the Mueller potential \cite{ren2002probing,muller1980reaction}.  The right panel shows the corresponding rugged energy landscape where the Mueller potential is perturbed by high frequency periodic fluctuations.}
\label{fig:potential}
\end{figure}

For the case where (\ref{eq:sde1}) possesses one characteristic lengthscale controlled by $\epsilon>0$, the convergence of $X_t^\epsilon$ to a coarse-grained process $X_t^0$ in the limit $\epsilon\rightarrow 0$ over a finite time interval is well-known.  When the rapid oscillations are periodic, under a diffusive rescaling this problem can be recast as a periodic homogenization problem, for which it can be shown that the process $X_t^\epsilon$ converges weakly to a Brownian motion with constant effective diffusion tensor $D$ (covariance matrix) which can be calculated by solving an appropriate Poisson equation posed on the unit torus,  see for example \cite{pavliotis2008multiscale,bensoussan1978asymptotic}.  The analogous case where the rapid fluctuations arise from a stationary ergodic random field has been studied in \cite[Ch. 9]{komorowski2012fluctuations}.  The case where the potential $V^\epsilon$ possesses periodic fluctuations with two or three well-separated characteristic timescales, i.e. $V^\epsilon(x) = V(x,x/\epsilon,x\epsilon^2)$ follow from the results in \cite[Ch. 3.7]{bensoussan1978asymptotic}, in which case the dynamics of the coarse-grained model in the $\epsilon\rightarrow 0$ limit are characterised by an It\^{o} SDE whose coefficients can be calculated in terms of the solution of an associated Poisson equation.  A generalization of these results to diffusion processes having $N$-well separated scales was explored in Section 3.11.3 of the same text, but no proof of convergence is offered in this case.   Similar diffusion approximations for systems with one fast scale and one slow scale,  where the fast dynamics are not periodic have been studied in \cite{pardoux2001poisson}.
\\\\
Further properties of the homogenized dynamics, in addition to the calculation of the mean first passage time, have been investigated.  For potentials of the form $V^\epsilon(x) = \alpha V(x) + p(x/\epsilon)$ for a smooth periodic function $p(\cdot)$ it was shown in \cite{pavliotis2007parameter} that the maximum likelihood estimator for the drift coefficients of the homogenized equation, given observations of the slow variable of the full dynamics (\ref{eq:sde1}) is \emph{asymptotically biased}.  Further results on inference of multiscale diffusions including (\ref{eq:sde1}) can be found in \cite{krumscheid2013semiparametric,krumscheid2014perturbation}.  In \cite{dupuis2011rare}, asymptotically optimal importance sampling schemes for studying rare events associated with (\ref{eq:sde1}) of the form $V^\epsilon(x) = V(x, x/\epsilon)$ were constructed by studying the $\epsilon\rightarrow 0$ limit of an associated Hamilton-Jacobi-Bellmann equation, the results were subsequently generalised to random stationary ergodic fluctuations in \cite{spiliopoulos2015rare}.  In \cite{Hartmann2014279}, the authors study optimal control problems for two-scale systems.   Small $\epsilon$ asymptotics for the exit time distribution of (\ref{eq:sde1}) were studied in \cite{almada2014scaling}.
\\\\
A model for Brownian dynamics in a potential $V$ possessing infinitely many characteristic lengthscales was studied in \cite{ben2003multiscale}.  In particular, the authors studied the large-scale diffusive behaviour of the overdamped Langevin dynamics in potentials of the form
\begin{equation}
\label{eq:arous_potential}
  V^n(x) = \sum_{k=0}^{n}U_k\left(\frac{x}{R_k}\right),
\end{equation}
obtained as a superposition of H\"{o}lder continuous periodic functions with period $1$.  It was shown in \cite{ben2003multiscale} that the effective diffusion coefficient decays exponentially fast with the number of scales, provided that the scale ratios $R_{k+1}/R_{k}$ are bounded from above and below, which includes cases where the is no scale separation.   From this the authors were able to show that the effective dynamics exhibits subdiffusive behaviour, in the limit of infinitely many scales.
\\\\
In this paper we study the dynamics of diffusion in a rugged potential possessing $N$ well-separated lengthscales.  More specifically, we study the dynamics of (\ref{eq:sde1}) where the multiscale potential is chosen to have the form
$$
  V^\epsilon(x) = V(x, x/\epsilon, x/\epsilon^2, \ldots, x/\epsilon^N),
$$
where $V$ is a smooth function, which is periodic in all but the first variable.  Clearly, $V$ can always be written in the form 
\begin{equation}
\label{eq:Veps}
  V(x_0, x_1, \ldots, x_N) = V_0(x_0) + V_1(x_0, x_1, \ldots, x_N),
\end{equation}
where  $(x_0,x_1,\ldots, x_N) \in \mathbb{R}^d\times \left(\mathbb{T}^d\right)^{N}$.  In this paper, we shall assume that the large scale component of the potential $V_0$ is smooth and confining on $\mathbb{R}^d$, and that the perturbation $V_1$ is a smooth bounded function which is periodic in all but the first variable.  Unlike \cite{ben2003multiscale}, we work under the assumption of explicit scale separation, however we also permit more general potentials than those of the form (\ref{eq:arous_potential}), allowing possibly nonlinear interactions between the different scales, and even full coupling between scales \footnote{we will refer to potentials of the form $V^\epsilon(x) = V_0(x) + V_1(x/\epsilon,\ldots, x/\epsilon^N)$ where $V_1$ is periodic in all variables as separable.}.  To emphasize the fact that the potential (\ref{eq:Veps}) leads to a fully coupled system across scales, we introduce the auxiliary processes $X_t^{(j)} = X_t/\epsilon^{j}$, $j=0,\ldots, N$.  The SDE (\ref{eq:sde1}) can then be written as a fully coupled system of SDEs driven by the same Brownian motion $W_t$,
\begin{subequations}
\label{eq:multiscale_system}
\begin{align}
dX^{(0)}_t &= -\sum_{i=0}^{N}\epsilon^{-i}\nabla_{x_i}V\left(X^{(0)}_t, X^{(1)}_t,\ldots, X^{(N)}_t\right)\,dt + \sqrt{2\sigma}\,dW_t\\
dX^{(1)}_t &= -\sum_{i=0}^{N}\epsilon^{-i+1}\nabla_{x_i}V\left(X^{(0)}_t, X^{(1)}_t,\ldots, X^{(N)}_t\right)\,dt + \sqrt{\frac{2\sigma}{\epsilon^2}}\,dW_t\\
 \notag&\vdots\\
dX^{(N)}_t &= -\sum_{i=0}^{N}\epsilon^{-i+N}\nabla_{x_i}V\left(X^{(0)}_t, X^{(1)}_t, \ldots, X^{(N)}_t\right)\,dt + \sqrt{\frac{2\sigma}{\epsilon^{2N}}}\,dW_t 
\end{align}
\end{subequations}
in which case $X_t^{(0)}$ is considered to be a ``slow'' variable, while $X_t^{(1)}, \ldots X_t^{(N)}$ are ``fast'' variables.  In this paper, we first provide an explicit proof of the convergence of the solution of (\ref{eq:sde1}), $X_t^\epsilon$  to a coarse-grained (homogenized) diffusion process $X_t^0$ given by the unique solution of the following It\^{o} SDE:
\begin{equation}
\label{eq:homogenized}
dX_t^0 = -\mathcal{M}(X_t^0)\nabla \Psi(X_t^0)\,dt + \sigma\nabla\cdot\mathcal{M}(X_t^0)\,dt + \sqrt{2\sigma\mathcal{M}(X_t^0)}\,dW_t,
\end{equation}
where 
$$\Psi(x) = -\sigma\log Z(x),$$
denotes the free energy, for 
$$
  Z(x) = 
  \int_{\mathbb{T}^d}\cdots\int_{\mathbb{T}^d}e^{- V(x, y_1, \ldots, y_N)/\sigma}\,dy_1\ldots dy_N,
$$
and where $\mathcal{M}(x)$ is a symmetric uniformly positive definite tensor 
which is independent of $\epsilon$.  The formula of the effective diffusion tensor is given in Section \ref{sec:setup}.   The multiplicative noise is due to the full coupling between the macroscopic and the $N$ microscopic scales.\footnote{For additive potentials of the form (\ref{eq:arous_potential}), i.e. when there is no interaction between the macroscale and the microscales, the noise in the homogenized equation is additive.} In particular, we show that although the noise in $X_t^\epsilon$ is additive,  the coarse-grained dynamics will exhibit multiplicative noise,  arising from the interaction between the microscopic fluctuations and the thermal fluctuations.   For one-dimensional potentials, we are able to obtain an explicit expression for $\mathcal{M}(x)$, regardless of the number of scales involved.  In higher dimensions, $\mathcal{M}(x)$ will be expressed in terms of the solution of a recursive family of Poisson equations which can be solved only numerically.  We also obtain a variational characterization of the effective diffusion tensor, analogous to the standard variational characterisations for the effective conductivity tensor for multiscale conductivity problems, see for example \cite{jikov2012homogenization}.  Using this variational characterisation, we are able to derive tight bounds on the effective diffusion tensor, and in particular, show that as $N \rightarrow \infty$, the eigenvalues of the effective diffusion tensor will converge to zero, suggesting that diffusion in potentials with infinitely many scales will exhibit anomalous diffusion.  The focus of this paper is the rigorous analysis of the homogenization problem for (\ref{eq:sde1}) with $V^\epsilon$ given by (\ref{eq:Veps}).  In a companion paper, \cite{duncan2016noise} we study in detail qualitative properties of the solution to the homogenized equation (\ref{eq:homogenized}), including noise-induced transitions and noise-induced hysteresis behaviour.
\\\\
For the cases $N=1,2$ the main result of this paper, namely the derivation of the coarse grained dynamics, arises as a special case of  \cite[Chapter 3.7]{bensoussan1978asymptotic}.  However, to our knowledge, the results in this paper are the first which rigorously prove the existence of this limit for arbitrarily many scales.  A standard tool for the rigorous analysis of periodic homogenization problems is two-scale convergence \cite{allaire1992homogenization,nguetseng1989general}.  This theory was extended to study reiterated homogenization problems in \cite{allaire1996multiscale}.  The techniques developed in these papers do not seem to be directly applicable to the problem here for several reasons:  first, we work in an unbounded domain, second the operators that we consider, i.e. the infinitesimal generator of the diffusion process (\ref{eq:sde1}) cannot  be written in divergence form.  The application of two-scale convergence to our problem would require extending two-scale convergence to weighted $L^2$-spaces,  that depend both on the large and small scale parameters, something which does not seem to be straightforward.   Our method for proving the homogenization theorem, Theorem~\ref{thm:homog_main} is based on the well known martingale approach to proving limit theorems~\cite{bensoussan1978asymptotic,papanicolaou1977martingale,pardoux2001poisson}. The main technical difficulty in applying such well known techniques is the construction of the corrector field/compensator. This turns out to be a very tedious task, since we consider the case where all scales, the macroscale and the $N$-- microscales, are fully coupled.
\\\\
Note that although we consider the homogenized process $X_t^0$, the solution of (\ref{eq:homogenized_sde}) to be a coarse grained version of the multiscale process $X_t^\epsilon$, both processes have the same configuration space.  We must therefore distinguish this approach with other coarse graining methodologies where effective dynamics are obtained for a lower dimensional set of coordinates of the original system, see for example \cite{legoll2010effective,blanc2010finite,hijon2010mori,peters2008projection}.  Nonetheless, one can still draw parallels between our approach and method described in \cite{legoll2010effective,blanc2010finite}.  Indeed, when writing (\ref{eq:sde1}) in the form  (\ref{eq:multiscale_system}) we can still view the limit $\epsilon \rightarrow 0$ as a form of dimension reduction,  approximating the fast-slow system (\ref{eq:multiscale_system}) of $N+1$ processes $(X^{(0)}_t, X_t^{(1)}, \ldots,  X_t^{(N)})$ taking values in $\mathbb{R}^{dN}$ by a single $\mathbb{R}^d$--valued process $X_t^0$ whose effective dynamics are characterised by the free energy $Z(x)$ and an effective diffusion tensor
% \\\\
% From the point of view of (\ref{eq:multiscale_system}), the process of homogenization reduces the $N+1$ system of coupled SDEs into a single slowly varying process $X_t^0$.
\\\\
Our assumptions on the potential $V^\epsilon$ in (\ref{eq:Veps}) guarantee that the full dynamics (\ref{eq:sde1}) is ergodic and reversible with invariant distribution $\pi^\epsilon$.  Furthermore, the coarse-grained dynamics (\ref{eq:homogenized}) is ergodic and reversible with respect to the equilibrium distribution 
$$\pi^0(x) = Z(x)/\overline{Z}.$$  
Indeed, the natural intepretation of $\Psi(x)=-\sigma \log Z(x)$ is as the free energy corresponding to the coarse-grained variable $X_t$.   The weak convergence of $X_t^\epsilon$ to $X_t^0$ implies in particular that the distribution of $X_t^\epsilon$ will converge weakly to that of $X_t^0$, uniformly over finite time intervals $[0,T]$,  which does not say anything about the convergence of the respective stationary distributions $\pi^\epsilon$ to $\pi^0$.  In Section \ref{sec:equilibrium} we study the equilibrium behaviour of $X_t^\epsilon$ and $X_t^0$ and show that the long-time limit $t\rightarrow \infty$ and the coarse-graining limit $\epsilon\rightarrow 0$ commute, and in particular that the equilibrium measure $\pi^\epsilon$ of $X_t^\epsilon$ converges in the weak sense to $\pi^0$.  We also study the rate of convergence to equilibrium for both processes, and we obtain bounds relating the two rates.  This question is naturally related to the study of the Poincar\'{e} constants for the full and coarse--grained potentials.
\\\\
The rest of the paper is organized as follows.  In Section \ref{sec:setup} we state the assumptions on the structure of the multiscale potential and state the main results of this paper.  In Section \ref{sec:properties} we study properties of the effective dynamics, providing expressions for the diffusion tensor in terms of a variational formula, and derive various bounds.  In Section \ref{sec:equilibrium}  we study properties of the effective potential, and prove convergence of the equilibrium distribution of $X_t^\epsilon$ to the coarse-grained equilibrium distribution $\pi^0$.

% The novelty of this work
% \begin{enumerate}
% \item We write down explicitly the form for the coarse grained process, which to our knowledge isn't written anywhere else.
% \item We study the long term behaviour of the coarse grained process, and compare it to the long term behaviour of the multiscale system (Questions: Is the homogenized system ergodic/reversible.  If the multiscale system satisfies an LSI then will the homogenized system also have exponential decay of entropy?)
% \item Do the following limits commute?
% $$\lim_{\epsilon \rightarrow 0}\lim_{T\rightarrow \infty}\mathbb{E}[f(X_T^\epsilon)] = \lim_{T\rightarrow \infty}\lim_{\epsilon \rightarrow 0}\mathbb{E}[f(X_T^\epsilon)].$$
% \end{enumerate}

% Extensions of the model
% \begin{enumerate}
% \item Consider a potential of the form $V^\epsilon = V_0(x) + \epsilon^\alpha V_1(x, \ldots)$.  For $\alpha \geq N$ this becomes an averaging problem, and for $1 \leq \alpha  < N$ the problem becomes a mixed homogenization--averaging problem.  \textbf{Cite respective section in \cite{bensoussan1978asymptotic}}.
% \item Locally stationary random potential (to my knowledge this has not been done).
% \item Time dependent potentials \cite{garnier1997homogenization}
% \end{enumerate}

\section{Setup and Statement of Main Results}
\label{sec:setup}
In this section we provide conditions on the multiscale potential which are required to obtain a well-defined homogenization limit.  In particular, we shall highlight assumptions necessary for the ergodicity of the full model as well as the coarse-grained dynamics.
\\\\
We will consider the overdamped Langevin dynamics
\begin{align}
\label{eq:sde_multiscale}
dX^\epsilon_t = -\nabla V^{\epsilon}(X_t^\epsilon)\,dt + \sqrt{2\sigma}\,dW_t,
\end{align}
where $V^\epsilon(x)$ is of the form
\begin{equation}
\label{eq:multiscale_potential}
  V^\epsilon(x) = V\left(x, \frac{x}{\epsilon}, \frac{x}{\epsilon^2}, \ldots, \frac{x}{\epsilon^N}\right),
\end{equation}
and where $V:\mathbb{R}^d\times\mathbb{T}^d\times\ldots\times \mathbb{T}^d \rightarrow \mathbb{R}$ is a smooth function which is assumed to be periodic with period $1$ in all but its first argument.   The multiscale potentials we consider in this paper can be viewed as a smooth confining potential perturbed by smooth, bounded fluctuations which become increasingly rapid as $\epsilon \rightarrow 0$, see Figure \ref{fig:potential} for an illustration.   More specifically, we will assume that the multiscale potential $V$ satisfies the following assumptions.\footnote{We remark that we can always write (\ref{eq:Veps}) in the form (\ref{eq:potential_split}) where $V_0(x) = \int_{\mathbb{T}^d}\cdots\int_{\mathbb{T}^d}V(x,x_1,\ldots, x_N)\,dx_1\,\ldots dx_N$. }
\begin{assumption}
\label{ass:potential1}
 The potential $V$ is given by
\begin{equation}
\label{eq:potential_split}
V(x_0, x_1, \ldots, x_N) = V_0(x_0) + V_1(x_0, x_1, \ldots, x_N),
\end{equation}
where:
\begin{enumerate}
\item $V_0$ is a smooth confining potential, i.e. $e^{-V_0(x)} \in L^1(\mathbb{R}^d)$ and  $V_0(x) \rightarrow \infty$ as $|x|\rightarrow \infty$.
\item The perturbation $V_1(x_0, x_1, \ldots, x_N)$ is smooth and bounded uniformly in $x$, independently of $\epsilon$.
\item  There exists $C > 0$ such that $\left\lVert \nabla^2 V_0 \right\rVert_{L^\infty(\mathbb{R}^d)} \leq C$.
\end{enumerate}
\end{assumption}
\begin{remark}
We note that Assumption 3 quite stringent, since it implies that $V_0$ is quadratic to leading order.  This assumption is also made in \cite{pardoux2001poisson}.  In cases where the process $X^{\epsilon}_0 \sim \pi^\epsilon$, i.e. the process is started in stationary, this condition can be relaxed considerably. 
\end{remark}
The infinitesimal generator $\mathcal{L}^\epsilon$ of $X_t^\epsilon$ is the  selfadjoint extension of
\begin{equation}
\label{eq:multiscale_generator}
  \mathcal{L}^\epsilon f(x) = -\nabla V^\epsilon(x)\cdot\nabla f(x) + \sigma\Delta f(x),\quad f\in C^\infty_0(\mathbb{R}^d).
\end{equation}
Since $V_0$ is confining, it follows that the corresponding overdamped Langevin equation
\begin{equation}
\label{eq:unperturbed_sde}
  dZ_t = -\nabla V_0(Z_t)\,dt + \sqrt{2\sigma}dW_t,
\end{equation}
is ergodic with unique stationary distribution
$$
  \pi_{ref}(x) = \frac{1}{Z}\exp(-V_0(x)/\sigma), \quad Z = \int_{\mathbb{R}^d}e^{-V_0(x)/\sigma}\,dx.
$$
Since $V_1$ is bounded uniformly, by Assumption \ref{ass:potential1}, it follows that the potential $V^\epsilon$ is also confining, and therefore  $X^\epsilon_t$ is ergodic, possessing a unique invariant distribution given by $\pi^\epsilon(x) = \frac{e^{-V^\epsilon(x)/\sigma}}{Z^\epsilon},$  where $Z^\epsilon = \int_{\mathbb{R}^d} e^{-V^\epsilon(x)/\sigma}$.  Moreover, noting that the generator $\mathcal{L}^\epsilon$ of $X_t^\epsilon$ can be written as
$$
  \mathcal{L}^\epsilon f(x) = \sigma\,e^{V^\epsilon(x)}\nabla\cdot\left(e^{-V^\epsilon(x)}\nabla f(x)\right), \quad f \in C^2_0(\mathbb{R}^d).
$$
it follows that $\pi^\epsilon$ is reversible with respect to the dynamics $X_t^\epsilon$, c.f. \cite{pavliotis2014stochastic,gardiner1985handbook}.\\\\
Our main objective in this paper is to study the dynamics (\ref{eq:sde_multiscale}) in the limit of infinite scale separation $\epsilon\rightarrow 0$. Having introduced the model and the assumptions we can now present the main result of the paper.

\begin{theorem}[Weak convergence of $X_t^\epsilon$ to $X^0_t$]
\label{thm:homog_main}
Suppose that Assumption \ref{ass:potential1} holds and let $T > 0$, and the initial condition $X_0$ is distributed according to some probability distribution $\nu$ on $\mathbb{R}^d$.  Then as $\epsilon\rightarrow 0$, the process $X_t^\epsilon$ converges weakly in $(C[0,T] ; \mathbb{R}^d)$ to the diffusion process $X_t^0$ with  generator defined by
\begin{equation}
\label{eq:homogenized_generator}
\mathcal{L}^0 f(x) = \frac{\sigma}{Z(x)}\nabla_x\cdot\left(Z(x)\mathcal{M}(x)\nabla_x f(X) \right),\quad f\in C^2_0(\mathbb{R}^d),
\end{equation}
and where
\begin{equation}
\label{eq:effective_potential}
Z(x) = \int_{\mathbb{T}^d}\cdots\int_{\mathbb{T}^d}e^{- V(x, x_1,\ldots, x_N)/\sigma}\,dx_N\ldots dx_1
\end{equation}
and 
\begin{equation}
\label{eq:effective_diffusion}
\mathcal{M}(x) = \frac{1}{Z(x)}\int_{\mathbb{T}^d}\cdots\int_{\mathbb{T}^d}(1 + \nabla_{x_N}\theta_N^\top)\cdots(1+\nabla_{x_1}\theta_1^\top)e^{- V(x,x_1,\ldots,x_N)/\sigma}\,dx_N\cdots dx_1. 
\end{equation}
The correctors are defined recursively as follows:   define $\theta_{N-k}$ to be the weak solution  of the PDE 
\begin{equation}
\label{eq:corrector}
  \nabla_{x_{N-k}}\cdot(\mathcal{K}_{N-k}(x_0, \ldots, x_{N-k})(\nabla_{x_{N-k}}\theta_{x_{N-k}}(x_0,\ldots, x_{N-k})+I)) =0,
\end{equation}
where $\theta_{N-k}(x_0,\ldots, x_{N-k-1},\cdot) \in H^1(\mathbb{T}^d)$ and where
\begin{equation}
\begin{aligned}
\mathcal{K}_{N-k}(&x_0,\ldots, x_{N-k})\\ = &\int_{\mathbb{T}^d}\cdots\int_{\mathbb{T}^d}(I + \nabla_N\theta_N^\top)\cdots(I+\nabla_{N-k+1}\theta_{N-k+1}^\top)e^{- V/\sigma}\,dx_N\ldots dx_{N-k+1},
\end{aligned}
\end{equation}
for $k = 1, \ldots, N-1$, and
\begin{equation}
\mathcal{K}_N(x, x_1,\ldots, x_N) = e^{- V(x, x_1,\ldots, x_N)/\sigma}I
\end{equation}
where $I$ denotes the identity matrix in $\mathbb{R}^{d\times d}$.  Provided that Assumptions \ref{ass:potential1} hold, Proposition \ref{prop:existence_poisson_general} guarantees existence and uniqueness (up to a constant) of solutions to the coupled Poisson equations (\ref{eq:corrector}).  Furthermore, the solutions will depend smoothly on the slow variable $x$ as well as the fast variables $y_1, \ldots, y_N$. The process $X_t^0$ is the unique solution to the It\^{o} SDE
\begin{equation}
\label{eq:homogenized_sde}
dX_t^0 = -\mathcal{M}(X_t^0)\nabla\Psi(X_t^0)\,dt +\sigma\nabla\cdot \mathcal{M}(X_t^0)\,dt+ \sqrt{2\sigma\mathcal{M}(X_t^0)}\,dW_t,
\end{equation}
where 
$$
\Psi(x) = -\sigma\log Z(x) = -\sigma\log\left(\int_{\mathbb{T}^d}\cdots\int_{\mathbb{T}^d} e^{-V(x, y_1, \ldots, y_N)/\sigma}\,dy_1\ldots dy_N\right).
$$
\end{theorem}
The proof, which closely follows that of \cite{pardoux2001poisson} is postponed to Section \ref{sec:homog_proof}.  Theorem \ref{thm:homog_main} confirms the intuition that the coarse-grained dynamics is driven by the free energy.  On the other hand, the corresponding SDE has multiplicative noise given by a space dependent diffusion tensor $\mathcal{M}(x)$. We can show that the homogenized process (\ref{eq:homogenized_sde}) is ergodic with unique invariant distribution
$$
  \pi^0(x) = \frac{Z(x)}{\overline{Z}} = \frac{1}{\overline{Z}}e^{-\Psi(x)/\sigma},\quad \mbox{ where } \quad \overline{Z} = \int_{\mathbb{R}^d}Z(x)\,dx.
$$
It is important to note that the reversibility of $X_t^\epsilon$ with respect to $\pi^\epsilon$ is preserved under the homogenization procedure.  In particular, the homogenized SDE (\ref{eq:homogenized_sde}) will be reversible with respect to the Gibbs measure $\pi^0(x)$.  Indeed, (\ref{eq:homogenized_sde}) has the form of the most general diffusion process that is reversible with respect to $\pi^0(x)$, see \cite[Sec. 4.7]{pavliotis2014stochastic}.
\\\\ 
While Theorem \ref{thm:homog_main} only characterises the convergence of $X_t^\epsilon$ to $X_t^0$ over finite time intervals, quite often we are interested in the equilibrium behaviour and in the rate of convergence to equilibrium for the coarse--grained process.  In Section \ref{sec:equilibrium} we study the properties of the invariant distributions $\pi^\epsilon$ and $\pi^0$ of $X_t^\epsilon$ and $X_t^0$, respectively.  In particular, we show that $\pi^\epsilon$ converges to $\pi^0$ in the sense of weak convergence of probability measures, and moreover characterise the rate of convergence to equilibrium for both $X_t^\epsilon$ and $X_t^0$ in terms of $\epsilon$, the parameter which measures scale separation. 
\\\\
As is characteristic with homogenization problems,  when $d=1$ we can obtain, up to quadratures, an explicit expression for the homogenized SDE.   In this case, we obtain explicit expressions for the correctors $\theta_1, \ldots, \theta_N$, so that the intermediary coefficients $\mathcal{K}_1, \ldots, \mathcal{K}_N$  can be expressed as 
$$
\mathcal{K}_i(x_0, x_1, \ldots, x_{i}) = \left(\int e^{V(x_0, x_1, \ldots, x_{i}, x_{i+1}, \ldots, x_N)/\sigma}\,dx_{i+1}\ldots dx_N\right)^{-1},\quad i=1,\ldots, N.
$$

\begin{proposition}[Effective Dynamics in one dimension]
When $d=1$, the effective diffusion coefficient $\mathcal{M}(x)$ in (\ref{eq:homogenized_sde}) is given by
\begin{equation}
\label{eq:eff_diff_1d}
  \mathcal{M}(x) = \frac{1}{Z_1(x)\widehat{Z}_1(x)},
\end{equation}
where 
$$
Z_1(x) = \int\cdots\int e^{-V_1(x, x_1,\ldots, x_N)/\sigma}\,dx_1\ldots dx_N,
$$
and 
$$
\widehat{Z}_1(x) = \int\cdots\int e^{V_1(x, x_1,\ldots, x_N)/\sigma}\,dx_1\ldots dx_N.
$$
\end{proposition}

Equation (\ref{eq:eff_diff_1d}) generalises the expression for the effective diffusion coefficient for a two-scale potential that was derived in \cite{zwanzig1988diffusion} without any appeal to homogenization theory.    In higher dimensions we will not be able to obtain an explicit expression for $\mathcal{M}(x)$, however we are able to obtain bounds on the eigenvalues of $\mathcal{M}(x)$.  In particular, we are able to show that (\ref{eq:eff_diff_1d}) acts as a lower bound for the eigenvalues of $\mathcal{M}(x)$.  
\begin{proposition}
\label{corr:unif_pos_def}
The effective diffusion tensor $\mathcal{M}$ is uniformly positive definite over $\mathbb{R}^d$.  In particular,
\begin{equation}
\label{eq:M_bounds}
0 < \,e^{-{osc}(V_1)/\sigma} \leq \frac{1}{Z_1(x)\widehat{Z}_1(x)} \leq e\cdot\mathcal{M}(x)e \leq 1,
\end{equation}
for all $e \in \mathbb{R}^d$ such that $|e|=1$, where
$$
{osc}(V_1) = \sup_{\substack{x\in\mathbb{R}^d, \\ y_1,\ldots, y_N \in \mathbb{T}^d}} V_1(x,y_1,\ldots, y_N) - \inf_{\substack{x\in\mathbb{R}^d, \\ y_1,\ldots, y_N \in \mathbb{T}^d}} V_1(x,y_1,\ldots, y_N)
$$
\end{proposition}
This result follows immediately from Lemmas \ref{lem:corr_existence_ellipticity} and \ref{lemma:variational} which are proved in Section \ref{sec:properties}.
% \begin{proof}
% The left inequality follows directly from Lemma \ref{lem:corr_existence_ellipticity}.  The right inequality follows from Lemma \ref{lemma:variational} by choosing $v_1 = v_2 = \ldots = v_N = 0$ in the equation (\ref{eq:variational_formulation}) in the case where $i= 2$.
% \end{proof}
\begin{remark}
The bounds in (\ref{eq:M_bounds}) highlight the two extreme possibilities for fluctuations occurring in the potential $V^\epsilon$.  The inequality $\frac{1}{Z_1(x)\widehat{Z}_1(x)} \leq e\cdot \mathcal{M}(x)e$  is attained when the multiscale fluctuations $V_1(x_0, \ldots, x_N)$ are constant in all but one dimension (e.g. the analogue of a layered composite material, \cite[Sec 5.4]{cioranescu2000introduction}, \cite[Sec 12.6.2]{pavliotis2008multiscale}).  In the other extreme, the inequality $e\cdot\mathcal{M}(x) e = 1$ is attained in the abscence of fluctuations, i.e. when $V_1 = 0$.  
\end{remark}

\begin{remark}
Clearly, the lower bound in (\ref{eq:M_bounds}) becomes exponentially small in the limit as $\sigma \rightarrow 0$.
\end{remark}

While Theorem \ref{thm:homog_main}  guarantees weak convergence of $X_t^\epsilon$ to $X_t^0$ in $C([0,T];\mathbb{R}^d)$ for fixed $T$,  it makes no claims regarding the convergence at infinity, i.e. of $\pi^\epsilon$ to $\pi^0$.  However, under the conditions of Assumption \ref{ass:potential1} we can show that $\pi^\epsilon$ converges weakly to $\pi^0$, so that the  $T\rightarrow \infty$ and $\epsilon\rightarrow 0$ limits commute, in the sense that:
$$
  \lim_{\epsilon \rightarrow 0}\lim_{T\rightarrow \infty}\mathbb{E}[f(X_T^\epsilon)] = \lim_{T\rightarrow \infty}\lim_{\epsilon \rightarrow 0}\mathbb{E}[f(X_T^\epsilon)],
$$
for all $f \in L^2(\pi_{ref})$.

\begin{proposition}[Weak convergence of $\pi^\epsilon$ to $\pi^0$]
\label{prop:weak_limit}
Suppose that Assumption \ref{ass:potential1} holds.  Then for all $f \in L^2(\pi_{ref})$,
\begin{equation}
\label{eq:error_weak}
  \int_{\mathbb{R}^d} f(x)\,\pi^{\epsilon}(dx) \rightarrow \int_{\mathbb{R}^d} f(x)\pi^0(dx),
\end{equation}
as $\epsilon\rightarrow 0$.
\end{proposition}
If Assumption \ref{ass:potential1} holds, then for every $\epsilon > 0$, the potential $V^\epsilon$ is confining, so that the process $X_t^\epsilon$ is ergodic.   If the ``unperturbed'' process defined by (\ref{eq:unperturbed_sde})
converges to equilibrium exponentially fast in $L^2(\pi_{ref})$, then so will $X_t^\epsilon$ and $X_t^0$.   Moreover, we can relate the rates of convergence of the three processes.

\begin{proposition}
\label{prop:poincare}
Suppose that Assumptions \ref{ass:potential1} holds and let $P_t$ be the semigroup associated with the dynamics (\ref{eq:unperturbed_sde}) and suppose that $\pi_{ref}(x) = \frac{1}{Z_0}e^{-V_0(x)/\sigma}$ satisfies Poincar\'{e}'s inequality with constant $\rho/\sigma$, i.e.
\begin{equation}
\label{eq:poincare_unperturbed}
  \mbox{Var}_{\pi_{ref}}(f) \leq \frac{\sigma}{\rho}\int |\nabla f(x)|^2\, \pi_{ref}(dx),\quad f \in H^1(\pi_{ref}),
\end{equation}
or equivalently
\begin{equation}
\label{eq:conv1}
  \mbox{Var}_{\pi_{ref}}\left(P_t f\right) \leq e^{-2\rho t/\sigma} \mbox{Var}_{\pi_{ref}}(f),\quad f \in L^2(\pi_{ref}),
\end{equation}
for all $t \geq 0$.  Let $P_t^\epsilon$ and $P_t^0$ denote the semigroups associated with the full dynamics (\ref{eq:sde_multiscale}) and homogenized dynamics (\ref{eq:homogenized_sde}), respectively.  Then for all $f \in L^2(\pi_{ref})$,
\begin{equation}
\label{eq:conv2}
  \mbox{Var}_{\pi^\epsilon}(P_t^\epsilon f) \leq e^{-2\gamma t/\sigma}\mbox{Var}_{\pi^\epsilon}(f),
\end{equation}
and
\begin{equation}
\label{eq:conv3}
  \mbox{Var}_{\pi^0}(P_t^0 f) \leq e^{-2{\widetilde{\gamma}}t/\sigma}\mbox{Var}_{\pi^0}(f).
\end{equation}
for  $\gamma = \rho \, e^{-2{osc} (V_1)/\sigma}$ and $\widetilde{\gamma} = \rho e^{-3{osc}(V_1)/\sigma}$. 
% $$
  % \mbox{osc}(V_1) = \sup_{\substack{x_0 \in \mathbb{R}^d \\ x_1,\ldots, x_N \in \mathbb{T}^d }}V_1(x_0, x_1, \ldots, x_N) -  \inf_{\substack{x_0 \in \mathbb{R}^d \\ x_1,\ldots, x_N \in \mathbb{T}^d }}V_1(x_0, x_1, \ldots, x_N).
% $$
\end{proposition}
The proofs Propositions \ref{prop:weak_limit} and \ref{prop:test_functions} will be deferred to Section \ref{sec:equilibrium}. 

\section{Properties of the Coarse--Grained Process}
\label{sec:properties}
In this section we study the properties of the coefficients of the homogenized SDE (\ref{eq:homogenized_sde}) and its dynamics. 

\subsection{Separable Potentials}
Consider the special case where the potential $V^\epsilon$ is \emph{separable}, in the sense that the fast scale fluctuations do not depend on the slow scale variable, i.e.
$$
  V(x_0, x_1, \ldots, x_N) = V_0(x_0) + V_1(x_1, x_2, \ldots, x_N).
$$
Then, it is clear from the construction of the effective diffusion tensor (\ref{eq:effective_diffusion}) that $\mathcal{M}(x)$ will not depend on $x \in \mathbb{R}^d$.  Moreover, since
$$
  Z(x) = \int_{\mathbb{T}^d}\cdots\int_{\mathbb{T}^d} e^{-\frac{V_0(x) + V_1(y_1, \ldots, y_N)}{\sigma}}\,dy_1\ldots dy_N = \frac{1}{K} e^{-V_0(x)/\sigma},
$$
where $K = \int_{\mathbb{T}^d}\cdots\int_{\mathbb{T}^d}\exp(-V_1(y_1,\ldots,y_N)/\sigma)\,dy_1\,\ldots dy_N$, then it follows that the coarse--grained stationary distribution $\pi^0$ equals the stationary distribution $\pi_{ref} \propto \exp(-V_0(x)/\sigma)$ of the process (\ref{eq:unperturbed_sde}).  For general multiscale potentials however,  $\pi^0$ will be different from $\pi_{ref}$.  Indeed, introducing multiscale fluctuations can dramatically alter the qualitative equilibrium behaviour of the process, including noise-inductioned transitions and noise induced hysteresis, as has been studied for various examples in \cite{duncan2016noise}.

\subsection{Variational bounds on $\mathcal{M}(x)$}
A first essential property is that the constructed matrices $\mathcal{K}_N, \ldots, \mathcal{K}_1$ are uniformly elliptic with respect to all their parameters, which is shown in the following lemma.  For convenience, we shall introduce the notation 
\begin{equation}
\label{eq:XK}
\mathbb{X}_k = \mathbb{R}^d \times \bigtimes_{i=1}^{k} \mathbb{T}^d
\end{equation}
for $k=1,\ldots, N$, and set $\mathbb{X}_0 = \mathbb{R}^d$ for consistency.  First we require the following existence and regularity result for a uniformly elliptic Poisson equation on $\mathbb{T}^d$.

\begin{lemma}
\label{lem:corr_existence_ellipticity}
For $k=1,\ldots, N$, the tensor $\mathcal{K}_k(x_0, \ldots, x_{k-1}, \cdot)$ is uniformly positive definite and in particular satisfies, for all unit vectors $e \in \mathbb{R}^d$,
  \begin{equation}
  \label{eq:K_ellipticity}
    \frac{1}{\widehat{Z}_k(x_0, x_1, \ldots, x_{k-1})} \leq  e\cdot\mathcal{K}_{k}(x_0, x_1, \ldots, x_{k-1}, x_k) \,e,\quad x_k \in \mathbb{T}^d
  \end{equation}
  where 
  $$
    \widehat{Z}_k(x_0, x_1, \ldots, x_{k-1}) = \int\ldots\int e^{ V(x_0, x_1, \ldots, x_{k-1}, x_k, \ldots, x_N)/\sigma}\,dx_N dx_{N-1}\ldots dx_{k},
  $$
  which is independent of $x_k$.  Moreover, the tensor $\mathcal{K}_k$ satisfies $\left(\mathcal{K}_k\right)_{i,j} \in C_b^\infty(\mathbb{X}_k)$, for all $i, j \in \lbrace 1, \ldots, d\rbrace$.
\end{lemma}
\begin{proof}
We prove the result by induction on $k$ starting from $k = N$.  For $k = N$ the tensor $\mathcal{K}_N$ is clearly uniformly positive definite for fixed $x_0, \ldots, x_{N-1}\in\mathbb{X}_{N-1}$.  The existence of the solution $\theta_N$ of (\ref{eq:corrector}) is then ensured by the Lemma \ref{lem:corr_existence_ellipticity}, and moreover it follows that $\mathcal{K}_{N-1}$ is well defined.  To show that $\mathcal{K}_{N-1}(x_0, \ldots, x_{N-2}, \cdot)$ is uniformly elliptic on $\mathbb{T}^d$ we first note that
\begin{equation}
\label{eq:K_expression2}
\begin{aligned}
  \int_{\mathbb{T}^d}(I + \nabla_{x_N} \theta_N)^\top &(I + \nabla_{x_N} \theta_N) e^{-V/\sigma}\,dx_N \\ &= \int_{\mathbb{T}^d} \left(I + \nabla_{x_N}\theta_N + \nabla_{x_N}\theta_N^\top +  \nabla_{x_N}\theta_N^\top\nabla_{x_N}\theta_N\right)e^{-V/\sigma}dx_N,
\end{aligned}
\end{equation}
where $V = V(x_0, x_1,\ldots, x_N)$, for $x_0,\ldots,x_{N-1} \in \mathbb{X}_{N-1}$ fixed, and where $\top$  denotes the transpose. From the Poisson equation for $\theta_N$ we have
$$
  \int \theta_N \cdot \nabla_{x_N}\cdot(e^{-V/\sigma}(\nabla_{x_N}\theta_N + I))e^{-V/\sigma}\,dx_N = 0,
$$
from which we obtain, after integrating by parts:
$$
  \int_{\mathbb{T}^d} \nabla_{x_N}\theta_N^\top \nabla_{x_N}\theta_N e^{-V/\sigma}\,dx_N = -\int \nabla_{x_N}\theta_N^\top e^{-V/\sigma}\,dx_N,
$$
so that
\begin{align*}
  \mathcal{K}_{N-1} &= \int_{\mathbb{T}^d} (I + \nabla_{x_N} \theta_N)^\top (I + \nabla_{x_N} \theta_N) e^{-V/\sigma}\,dx_N.
\end{align*}
We note that
$$
  \int_{\mathbb{T}^d} (I + \nabla_N \theta_N)\,dx_N = I,
$$
therefore, it follows by H\"{o}lder's inequality that
$$
  |v|^2 \leq  \left|v\cdot \int_{\mathbb{T}^d} (I + \nabla_N \theta_N)v\right|^2 \leq v\cdot\left(\mathcal{K}_{N-1} \right)v \int_{\mathbb{T}^d} e^{V/\sigma}\,dx_N,
$$
so that
$$
  \frac{|v|^2}{\widehat{Z}_N(x_0, \ldots, x_{N-1})} \leq v\cdot\mathcal{K}_{N-1}(x_0,\ldots, x_{N-1}) v,\quad \forall (x_0, x_1, \ldots, x_{N_1}).
$$
Since $V_1$ is uniformly bounded over $x_0, \ldots, x_{N-1}$ it follows that $\widehat{Z}_N$ is strictly positive, so  that $\mathcal{K}_{N-1}$ is uniformly elliptic, and arguing as above we obtain existence of a unique $\theta_{N-1}$, up to a constant, solving (\ref{eq:theta_k}) for $k=2$.
\\\\
Now, assume that the correctors have been constructed for $i=N, \ldots, N-k+1$ and consider the tensor
\begin{equation}
\label{eq:tensor_nk1}
 \begin{aligned}
    \int\cdots\int  &(I+\nabla_{i+1}\theta_{i+1})^\top\cdots (I+\nabla_{k+1}\theta_{k+1})^\top  \\
        &\mathcal{K}_k  (I+\nabla_{k+1}\theta_{k+1})\cdots(I+\nabla_{i+1}\theta_{i+1})\, dx_N\ldots dx_{i+1}.
  \end{aligned}
\end{equation}
  Integrating by parts the cell equation for $\theta_{k+1}$  we see that
  $$
    \int \left(I + \nabla_{k+1}\theta_{k+1}\right)^\top \mathcal{K}_k \left(I + \nabla_{k+1}\theta_{k+1}\right)\,dx_{k+1} = \mathcal{K}_{k-1}.
  $$
  Continuining this approach by induction, it follows that (\ref{eq:tensor_nk1}) equals $\mathcal{K}_{i+1}$, thus proving the  representation (\ref{eq:K_expression2}), as required.  We now verify (\ref{eq:K_ellipticity}).  First we note that
  $$
    \int\cdots \int (I + \nabla_{N}\theta_N) \cdots (I + \nabla_{i+1}\theta_{i+1}) dx_N\ldots dx_{i+1} = I.
  $$
  Therefore, for any vector $v \in \mathbb{R}^d$:
  \begin{align*}
    |v|^2 &\leq   \left|\left(\int\cdots \int (I + \nabla_{N}\theta_N) \cdots (I + \nabla_{i+1}\theta_{i+1}) dx_N\ldots dx_{i+1}\right)v \right|^2
    \\ &\leq v\cdot  \left( \int\cdots\int  (I+\nabla_{i+1}\theta_{i+1})^\top\cdots (I+\nabla_{i+1}\theta_{i+1})e^{-V/\sigma}dx_N\ldots dx_{i+1}\right)v \int e^{V/\sigma} dx_N\ldots dx_{i+1} \\
    & = \left(v\cdot \mathcal{K}_{i+1}(x_1, \ldots, x_i)v\right) \widehat{Z}(x_1, \ldots, x_i).
  \end{align*}
  The fact that we have strict positivity for fixed $x_1, \ldots x_i$ then follows immediately.
\end{proof}

To obtain upper bounds for the effective diffusion coefficient, we will express the intermediary diffusion tensors $\mathcal{K}_i$ as solutions of a quadratic variational problem.  This variational formulation of the diffusion tensors can be considered as a generalisation of the analogous representation for the effective conductivity coefficient of a two-scale composite material, see for example \cite{jikov2012homogenization,milton1995theory,bensoussan1978asymptotic}.  
\begin{lemma} 
\label{lemma:variational}
For $i=1,\ldots, N$, the tensor $\mathcal{K}_{i}$ satisfies
\begin{equation}
\label{eq:variational_formulation}
\begin{aligned}
&e \cdot \mathcal{K}_{i}(x_0,\ldots, x_{i})e\\ 
 &= \inf_{v_{i+1}, \ldots, v_N \in H^1(\mathbb{T}^d)}\int_{(\mathbb{T}^d)^N} \left|e + \nabla v_{i+1}(x_i) + \ldots + \nabla v_N(x_N)\right|^2  e^{-V(x_0,\ldots, x_{N})/\sigma}\,dx_N\ldots dx_{i+1},
 \end{aligned}
\end{equation}
for all $e \in \mathbb{R}^d$.
\end{lemma}
\begin{proof}
  For $i=1,\ldots, N$, from the proof of Lemma \ref{lem:corr_existence_ellipticity} we can express the intermediary diffusion tensor $\mathcal{K}_{i-1}$ in the following recursive manner,
  \begin{align*}
  e\cdot &\mathcal{K}_{i-1}(x_0, \ldots, x_{i-1})e \\
  &= \int_{\mathbb{T}^d}(e+ e\cdot \nabla_{x_{i}}\theta_{i}(x_0,\ldots, x_i))^\top \mathcal{K}_{i}(x_0,\ldots, x_i)(e+e\cdot\nabla_{x_{i}}\theta_{i}(x_0,\ldots, x_i))\, d x_{i}.
  \end{align*}
  For fixed $x_0,\ldots, x_{i-1} \in \mathbb{X}_{i-1}$ and $e \in \mathbb{R}^d$,  consider the tensor $\widetilde{\mathcal{K}}_{i-1}$ defined by the following   quadratic minimization problem
  \begin{equation}
  \label{eq:K_tilde}
    e\cdot \widetilde{\mathcal{K}}_{i-1}(x_0, \ldots, x_{i-1})e = \inf_{v \in H^1(\mathbb{T}^d)} \int_{\mathbb{T}^d} (e + \nabla v(x_i))\cdot \mathcal{K}_i(x_0, \ldots, x_i)(e + \nabla v(x_i))\,dx_i.
  \end{equation}
  Since $\mathcal{K}_i$ is a symmetric tensor, the corresponding Euler-Lagrange equation for the minimiser is given by
  $$
    \nabla_{x_i}\cdot\left(\mathcal{K}_i(x_0, \ldots, x_i)(\nabla_{x_i}\chi(x_0, \ldots, x_i) + e)\right)= 0, \quad x \in \mathbb{T}^d,
  $$
  with periodic boundary conditions. This equation has unique mean zero solution given by $\chi(x_0, \ldots, x_i) = \theta_i(x_0,\ldots, x_i)\cdot e$, where $\theta_i$ is the unique mean-zero solution of (\ref{eq:corrector}).  It thus follows that $e\cdot \mathcal{K}_{i-1}e = e\cdot \widetilde{\mathcal{K}}_{i-1}e$, where $\widetilde{K}_{i-1}$ is given by (\ref{eq:K_tilde}).   Expanding $\mathcal{K}_{i}$ in a similar fashion, we obtain
  \begin{align*}
 & e\cdot \mathcal{K}_{i-1}(x_0, \ldots, x_{i-1})e\\
  & = \inf_{v_i, v_{i+1} \in H^1(\mathbb{T}^d)} \int_{\mathbb{T}^d} \int_{\mathbb{T}^d} (e + \nabla v_i(x_i) + \nabla v_{i+1}(x_{i+1}))\cdot \mathcal{K}_{i+1}(x_0, \ldots, x_{i+1})(e + \nabla v_i(x_i) + \nabla v_{i+1}(x_{i+1}))\,dx_{i+1}dx_i.
  \end{align*}
  Proceeding recursively, we arrive at
  \begin{align*}
  &e\cdot \mathcal{K}_{i-1}(x_0, \ldots, x_{i-1})e \\ &= \inf_{v_i,\ldots, v_N \in H^1(\mathbb{T}^d)}\int_{(\mathbb{T}^d)^{N}} \left|e + \nabla v_i(x_i) + \ldots + \nabla v_{N}(x_{N})\right|^2 e^{-V(x_0,\ldots, x_N)/\sigma}\,dx_N\ldots, dx_i, 
  \end{align*}
  as required.
\end{proof}

\begin{remark}
Proposition \ref{corr:unif_pos_def} follows immediately from Lemma \ref{lemma:variational} by choosing $$v_1 = v_2 = \ldots = v_N = 0,$$ in (\ref{eq:variational_formulation}) in the case where $i= 1$.
\end{remark}
% \begin{remark}
% Suppose we wish to simulate an overdamped Langevin equation for a multiscale potential $V^\epsilon = V(x, x/\epsilon, \ldots, x/\epsilon^N)$.   An ad-hoc means of eliminating the smaller scale fluctuations would be to consider the corresponding Langevin equation with the higher-order fluctuations ``integrated out''.  More specifically, we would consider the overdamped Langevin process $\widetilde{X}_t^\epsilon$ given by
% $$
% dX_t^\epsilon = -\nabla \widetilde{V}^\epsilon(X_t)\,dt + \sqrt{2\sigma}\,dW_t,
% $$
% where
% $$
%   \widetilde{V}^\epsilon(x) = -\sigma \log \int_{\mathbb{T}^d}e^{-V(x, x/\epsilon, \ldots, x/\epsilon^{N-1}, y)/\sigma}\,dy.
% $$
% Then the invariant distribution of this process will still be $\pi^0(x) = Z(x)/\overline{Z}$, however the resulting effective diffusion coefficient $\widetilde{\mathcal{M}}(x)$ will be different from $\mathcal{M}(x)$ given in (\ref{eq:effective_diffusion}), and moreover from Lemma \ref{lemma:variational}, it follows that
% $$
%   e\cdot \mathcal{M}(x) e \leq e \cdot\widetilde{\mathcal{M}}(x) e,
% $$
% for all $e \in \mathbb{R}^d$.
% \end{remark}

% An immediately corollary of the previous two lemmas is the following bounds on the eigenvalues of the matrix $\mathcal{M}(x)$.

\section{Properties of the Equilibrium Distributions}
\label{sec:equilibrium}
In this section we study in more detail the properties of the equilibrium distributions $\pi^\epsilon$ and $\pi^0$ of the full (\ref{eq:sde_multiscale}) and homogenized (\ref{eq:homogenized_sde}) dynamics, respectively.  We first provide a proof of Proposition \ref{prop:weak_limit}.  The approach we follow in this proof is based on properties of periodic functions, in a manner similar to \cite[Sec. 2]{cioranescu2000introduction}.  

\begin{proof}[Proof of Proposition \ref{prop:weak_limit}]
  First we note that, by Assumptions \ref{ass:potential1}, there exists a $C > 0$ independent of $\epsilon$, such that
  $$
    \int_{\mathbb{R}^d} \left|e^{-V_1(x, x/\epsilon, \ldots, x/\epsilon^N)/\sigma}\right|^2 e^{-V_0(x)/\sigma}\,dx \leq C < \infty.
  $$
  It follows that there exists $\Lambda \in L^2(\mathbb{R}^d; e^{-V_0/\sigma})$ and a subsequence $\left(\epsilon_n\right)_{n \in \mathbb{N}}$ where $\epsilon_n \rightarrow 0$ such that
  $$
    \int_{\mathbb{R}^d} e^{-V_1(x, x/{\epsilon}_n, \ldots, x/{\epsilon_n}^N)/\sigma} g(x)e^{-V_0(x)/\sigma}\,dx \xrightarrow{n\rightarrow \infty} \int_{\mathbb{R}^d} \Lambda(x) g(x)e^{-V_0(x)/\sigma}\,dx,
  $$
  for all $g \in L^2(\pi_{ref})$.  To identify the limit, we choose $g = \mathbf{1}_{\Omega}$ where $\Omega$ is an open bounded subset of $\mathbb{R}^d$ where $\partial \Omega$ is smooth; noting that the span of such functions is dense in $L^2(\pi_{ref})$.
  \\\\
  Following \cite{radu1992homogenization} and \cite[Sec. 2.3]{cioranescu2000introduction}, given $\Omega$ and $\epsilon > 0$, let $\lbrace Y_k \rbrace_{k=1,\ldots, N(\epsilon)}$ be a  collection of pairwise disjoint translations of $\mathbb{T}^d$, such that
  $\epsilon^N Y_k \subset \Omega$, for $k=1,\ldots, N(\epsilon)$ and for all $\delta > 0$, there exists $\epsilon_0$ such that
  $$
    \lambda\left({\Omega\setminus \cup_{k=1}^{N(\epsilon)}\epsilon^N Y_k}\right) < \delta,
  $$
  for all $\epsilon < \epsilon_0$, where $\lambda(\cdot)$ denotes the Lebesgue measure on $\mathbb{R}^d$.  Given $\delta > 0$, there exists $\epsilon_0$ such that for $\epsilon < \epsilon_0$,
  \begin{align*}
    \int_{\Omega} e^{-V^\epsilon(x)/\sigma}\,dx &= \sum_{k=1}^{N(\epsilon)}\int_{\epsilon^{N}Y_k}e^{-V^\epsilon(x)/\sigma}\,dx + O(\delta)\\ &=  \sum_{k=1}^{N(\epsilon)}\int_{\epsilon^N(x_k +  \mathbb{T}^d)} e^{-V(x, x/\epsilon, \ldots,  x/\epsilon^{N-1}, x/\epsilon^{N})/\sigma}\,dx + O(\delta)\\
    &=  \epsilon^{Nd}\sum_{k=1}^{N(\epsilon)}\int_{\mathbb{T}^d} e^{-V(\epsilon^N (x_k+y), \epsilon^{N-1}(x_k+y), \ldots, \epsilon (x_k+y), y)/\sigma}\,dy + O(\delta)\\
     &=  \epsilon^{Nd}\sum_{k=1}^{N(\epsilon)}\int_{\mathbb{T}^d} e^{-V(\epsilon^N x_k, \epsilon^{N-1}x_k, \ldots, \epsilon x_k, y)/\sigma}\,dy  + O(\delta)\\
     &= \int_{\cup_{k=1}^{N(\epsilon)}Y_k}\int_{\mathbb{T}^d} e^{-V(x, x/\epsilon, \ldots, x/\epsilon^{N-1}, y)/\sigma}\,dy\,dx  + O(\delta) \\
     &= \int_{\Omega}\int_{\mathbb{T}^d} e^{-V(x, x/\epsilon, \ldots, x/\epsilon^{N-1}, y)/\sigma}\,dy\,dx  + O(\delta),
  \end{align*}
  where we use the fact that $V$ is smooth with bounded derivatives on $\Omega$.  Proceeding iteratively in the above manner, we obtain that for all $\delta > 0$, there exists $\epsilon_0$ such that
  $$
  \int_{\Omega} e^{-V^\epsilon(x)/\sigma}\,dx = \int_{\Omega}\int_{\mathbb{T}^d}\cdots\int_{\mathbb{T}^d} e^{-V(x, y_1,\ldots, y_N)/\sigma}\,dy_N\,\ldots\,dy_N\,dx + O(\delta),
  $$
  for all $\epsilon < \epsilon_0$.  Thus it follows that 
  $$
    \Lambda(x) = \int_{\mathbb{T}^d}\cdots \int_{\mathbb{T}^d} e^{-V_1(x, y_1, \ldots, y_N)/\sigma}\,dy_N\,dy_{N-1}\,\ldots\,dy_1.
  $$
  In particular,
  $$
  Z^\epsilon = \int_{\mathbb{R}^d} e^{-V^\epsilon(x)/\sigma}\,dx \rightarrow Z^0 = \int_{\mathbb{R}^d}\int_{\mathbb{T}^d}\cdots \int_{\mathbb{T}^d} e^{-V(x, y_1,\ldots, y_N)/\sigma}\,dy_N\,\ldots\,dy_1\,dx,
  $$
  and thus, for all $h \in L^2(\mathbb{R}^d; e^{-V_0(x)/\sigma})$ 
  $$
    \int h(x)\pi^\epsilon(x)\,dx \rightarrow \int h(x)\pi^0(x)\,dx,
  $$
  as $\epsilon\rightarrow 0$, as required.
\end{proof}

\begin{proof}[Proof of Proposition \ref{prop:poincare}]
Since $V_1$ is bounded uniformly by Assumption \ref{ass:potential1}, it is straightforward to check that
\begin{equation}
  \pi_{ref}(x)e^{-2osc(V_1)/\sigma} \leq \pi^\epsilon(x) \leq \pi_{ref}(x)e^{2osc(V_1)/\sigma}. 
\end{equation}
It thus follows directly from (\ref{eq:poincare_unperturbed}), or alternatively from  \cite[Lemma 5.1.7]{bakry2013analysis}, that $\pi^\epsilon$ satisfies Poincar\'{e}'s inequality with constant $$\gamma = \frac{\rho}{\sigma} e^{-2\mbox{osc}(V_1)/\sigma},$$  which implies (\ref{eq:conv2}).  An identical argument follows for the coarse--grained density $\pi^0(x)$.  Finally, using the fact that
$$
|v|^2 e^{-{osc}(V_1)/\sigma} \leq \frac{|v|^2}{Z(x)\widehat{Z}(x)} \leq  v\cdot \mathcal{M}(x)v,
$$ 
for all $v \in \mathbb{R}^d$,  we obtain
\begin{align*}
  \mbox{Var}_{\pi^0}(f) &\leq \frac{\sigma}{\rho} e^{2{osc}(V_1)/\sigma}\int_{\mathbb{R}^d} |\nabla f(x) |^2 \,\pi^0(x)\,dx \\&\leq \frac{\sigma}{\rho}e^{3{osc}(V_1)/\sigma}\int \nabla f(x)\cdot \mathcal{M}(x) \nabla f(x)\, \pi^0(x)\,dx,
\end{align*}
from which (\ref{eq:conv3}) follows.
\end{proof}

\begin{remark}
Note that one can similarly relate the constants in the Logarithmic Sobolev inequalities for the measures $\pi_{ref}$, $\pi^\epsilon$ and $\pi^0$ in an almost identical manner, based on the Holley-Stroock criterion \cite{holley1987logarithmic}.
\end{remark}
\begin{remark}
Proposition \ref{prop:poincare} requires the assumption that the multiscale perturbation $V_1$ is bounded uniformly.  If this is relaxed, then it is no longer guaranteed that $\pi^\epsilon$ will satisfy a Poincar\'{e} inequality, even though $\pi_{ref}$ does. For example, consider the potential
$$
  V^\epsilon(x) = x^2(1 + \alpha \cos(2\pi x/\epsilon)),
$$
then the corresponding Gibbs distribution $\pi^\epsilon(x)$ will not satisfy Poincar\'{e}'s inequality for any $\epsilon > 0$.  Following  \cite[Appendix A]{hebisch2010coercive} we demonstrate this by checking that this choice of $\pi^\epsilon$ does not satisfy the Muckenhoupt criterion \cite{muckenhoupt1972hardy,ane2000inegalites} which is necessary and sufficient for the Poincar\'{e} inequality to hold, namely that $\sup_{r \in \mathbb{R}}B_{\pm}(r) < \infty$, where
$$
  B_{\pm}(r) = \left(\int_{r}^{\pm \infty} \pi^\epsilon(x)\,dx\right)^{\frac{1}{2}}\left(\int_{[0, \pm r]} \frac{1}{\pi^\epsilon(x)}\,dx\right)^{\frac{1}{2}}.
$$
Given $n \in \mathbb{N}$, we set $r/\epsilon = 2\pi n + \pi/2$.  Then we have that
\begin{align*}
  B_{+}(r) &\geq \left(\int_{\epsilon(2\pi n + 2\pi/3)}^{\epsilon(2\pi n + 4\pi/3)}e^{-|x|^2(1 -\alpha/2)/\sigma}\,dx\right)^{1/2}\left(\int_{\epsilon(2\pi n - \pi/3)}^{\epsilon(2\pi n + \pi/3)}e^{|x|^2(1 +\alpha/2)/\sigma}\,dx\right)^{1/2} \\
  &\geq \left(\frac{ 2\pi \epsilon}{3}\right)\exp\left(-\frac{|\pi\epsilon(2n + 4/3)|^2}{2\sigma}\left(1 - \frac{\alpha}{2}\right) + \frac{|\pi\epsilon(2n-1/3)|^2}{2\sigma}\left(1 + \frac{\alpha}{2}\right)\right) \\
  &= \left(\frac{ 2\pi \epsilon}{3}\right)\exp\left(-\frac{|2\pi\epsilon n|^2(1 + 2/3n)^2}{2\sigma}\left(1 - \frac{\alpha}{2}\right) + \frac{|2\pi\epsilon n|^2(1-1/6n)^2}{2\sigma}\left(1 + \frac{\alpha}{2}\right)\right)\\
  &\approx \left(\frac{ 2\pi \epsilon}{3}\right)\exp\left(\frac{|2\pi\epsilon n|^2}{2\sigma}\left(\alpha + o(n^{-1})\right)\right) \rightarrow \infty, \quad \mbox{ as  }n \rightarrow \infty,
\end{align*}
so that Poincar\'{e}'s inequality does not hold for $\pi^\epsilon$.
\end{remark}

$ $
A natural question to ask is whether the weak convergence of $\pi^\epsilon$ to $\pi^0$ holds true in a stronger notion of distance such as total variation.  The following simple one-dimensional example demonstrates that the convergence cannot be strengthened to total variation.

\begin{example}
Consider the one dimensional Gibbs distribution 
$$\pi^\epsilon(x) = \frac{1}{Z^\epsilon}e^{-V^\epsilon(x)/\sigma},$$
where
$$
  V^\epsilon(x) = \frac{x^2}{2} + \alpha \sin\left(2\pi \frac{x}{\epsilon}\right),
$$
and where $Z^\epsilon$ is the normalization constant and  $\alpha\neq 0$.  Then the measure $\pi^\epsilon$ converges weakly to $\pi^0$ given by
$$
  \pi^0(x) = \frac{1}{\sqrt{2\pi \sigma}}e^{-x^2/2\sigma}.
$$
From the plots of the  stationary distributions in Figure \ref{fig:one_dimensional_cex} it becomes clear that the density of $\pi^\epsilon$ exhibits rapid fluctuations which do not appear in $\pi^0$, thus we do not expect to be able to obtain convergence in a stronger metric.  First we consider the distance between $\pi^\epsilon$ and $\pi^0$ in total variation \footnote{we are using the same notation for the measure and for its density with respect to the Lebesgue measure on $\mathbb{R}^d$.}
\begin{align*}
\lVert \pi^\epsilon - \pi^0 \rVert_{TV} &= \int_{\mathbb{R}^d}|\pi^\epsilon(x) - \pi^0(x) |\,dx = \int_{\mathbb{R}^d}\frac{e^{-x^2/2\sigma}}{\sqrt{2\sigma}}\left|1 - \frac{e^{-\frac{\alpha}{\sigma} \cos(2\pi x/\epsilon)}}{K^\epsilon}\right|\,dx,
\end{align*}
where $K^\epsilon = Z^\epsilon/\sqrt{2\pi\sigma}$.  It follows that
\begin{align*}
\lVert \pi^\epsilon - \pi^0 \rVert_{TV} &\geq \sum_{n\geq 0}\int_{\epsilon(2\pi n - \pi/3)}^{\epsilon(2\pi n + \pi/3)} \frac{e^{-x^2/2\sigma}}{\sqrt{2\pi\sigma}}\,dx\left|1 -  \frac{e^{-\frac{\alpha}{2\sigma}}}{K^\epsilon} \right|  \\
& \geq \sum_{n\geq 0}\frac{2\epsilon \pi}{3}\frac{e^{-\epsilon^2( 2n\pi + \pi/3)^2/2\sigma}}{\sqrt{2\pi\sigma}}\left|1 -  \frac{e^{-\frac{\alpha}{2\sigma}}}{K^\epsilon} \right| \\
&\geq \int_{0}^{\infty} \frac{2\pi}{3}\frac{e^{-2\pi^2 (x + \epsilon/6)^2 /\sigma}}{\sqrt{2\pi\sigma}}\left|1 -  \frac{e^{-\frac{\alpha}{2\sigma}}}{K^\epsilon} \right|,
\end{align*}
where we use the fact that $e^{-\alpha/2\sigma}/K^\epsilon \leq 1$.  In the limit $\epsilon \rightarrow 0$, we have  $K^\epsilon \rightarrow I_0(\alpha/\sigma)$, where $I_n(\cdot)$ is the modified Bessel function of the first kind of order $n$.  Therefore, as $\epsilon \rightarrow 0$,
\begin{equation}
\lVert \pi^\epsilon - \pi^0 \rVert_{TV}\geq \int_{0}^{\infty} \frac{2\pi}{3}\frac{e^{-2\pi^2 (x + \epsilon/6)^2 /\sigma}}{\sqrt{2\pi\sigma}}\left|1 -  \frac{e^{-\frac{\alpha}{2\sigma}}}{K^\epsilon} \right| = \frac{1}{6}\left|1 -  \frac{e^{-\frac{\alpha}{2\sigma}}}{I_0(\alpha/\sigma)} \right|,
\end{equation}
which converges to $\frac{1}{6}$ as $\frac{\alpha}{\sigma}\rightarrow \infty$.  Since relative entropy controls total variation distance by Pinsker's theorem, it follows that $\pi^\epsilon$ does not converge to $\pi^0$ in relative entropy, either.  Nonetheless, we shall compute the distance in relative entropy  between $\pi^\epsilon$ and $\pi^0$ to understand the influence of the parameters $\sigma$ and $\alpha$.  Since both $\pi^0$ and $\pi^\epsilon$ have strictly positive densities with respect to the Lebesgue measure on $\mathbb{R}$, we have that
$$
  \frac{d\pi^\epsilon}{d\pi^0}(x) = \frac{\sqrt{2\pi\sigma}}{Z^\epsilon}e^{-V^\epsilon(x)/\sigma+\frac{1}{2}x^2/\sigma}.
$$
Then, for $Z^0 = \sqrt{2\pi \sigma}I_0(1/\sigma)$,
\begin{align*}
H\left(\pi^\epsilon \, | \, \pi^0 \right) &= \frac{1}{Z^\epsilon}\int \left(\frac{1}{2}\log(2\pi \sigma)-\log Z^\epsilon \right)e^{-V^\epsilon(x)/\sigma}\,dx \\ &\qquad + \frac{1}{Z^\epsilon}\int \left(-V^\epsilon(x)/\sigma + x^2/2\sigma\right)e^{-V^\epsilon(x)/\sigma}\,dx\\
&\xrightarrow{\epsilon \rightarrow 0} - \log I_0(\alpha/\sigma) + \frac{\alpha}{\sigma Z^0}\lim_{\epsilon\rightarrow 0}\int \cos(2\pi x/\epsilon)e^{-x^2/2\sigma - \alpha \cos(2\pi x/\epsilon)/\sigma}\,dx \\
&= -\log I_0(\alpha/\sigma) + \frac{\alpha}{\sigma} \frac{I_1(\alpha/\sigma)}{ I_0(\alpha/\sigma)} =: K(\alpha/\sigma).
\end{align*}
and it is straightfoward to check that  $K(s) > 0$, and moreover
$$
  K(s)   \rightarrow \begin{cases} +\infty & \mbox{as } s\rightarrow 0, \\ 0 & \mbox{ as } s\rightarrow \infty. \end{cases}
$$
In Figure \ref{fig:one_dimensional_k} we plot the value of $K(s)$ as a function of $s$.  From this result, we see that for fixed $\alpha > 0$, the measure $\pi^\epsilon$ will converge in relative entropy only in the limit as $\sigma \rightarrow \infty$,  while the measures will become increasingly mutually singular as $\sigma \rightarrow 0$.

\begin{figure}[h!]
    \centering
    \begin{subfigure}[b]{0.45\textwidth}
        \centering
    \includegraphics[width=\textwidth]{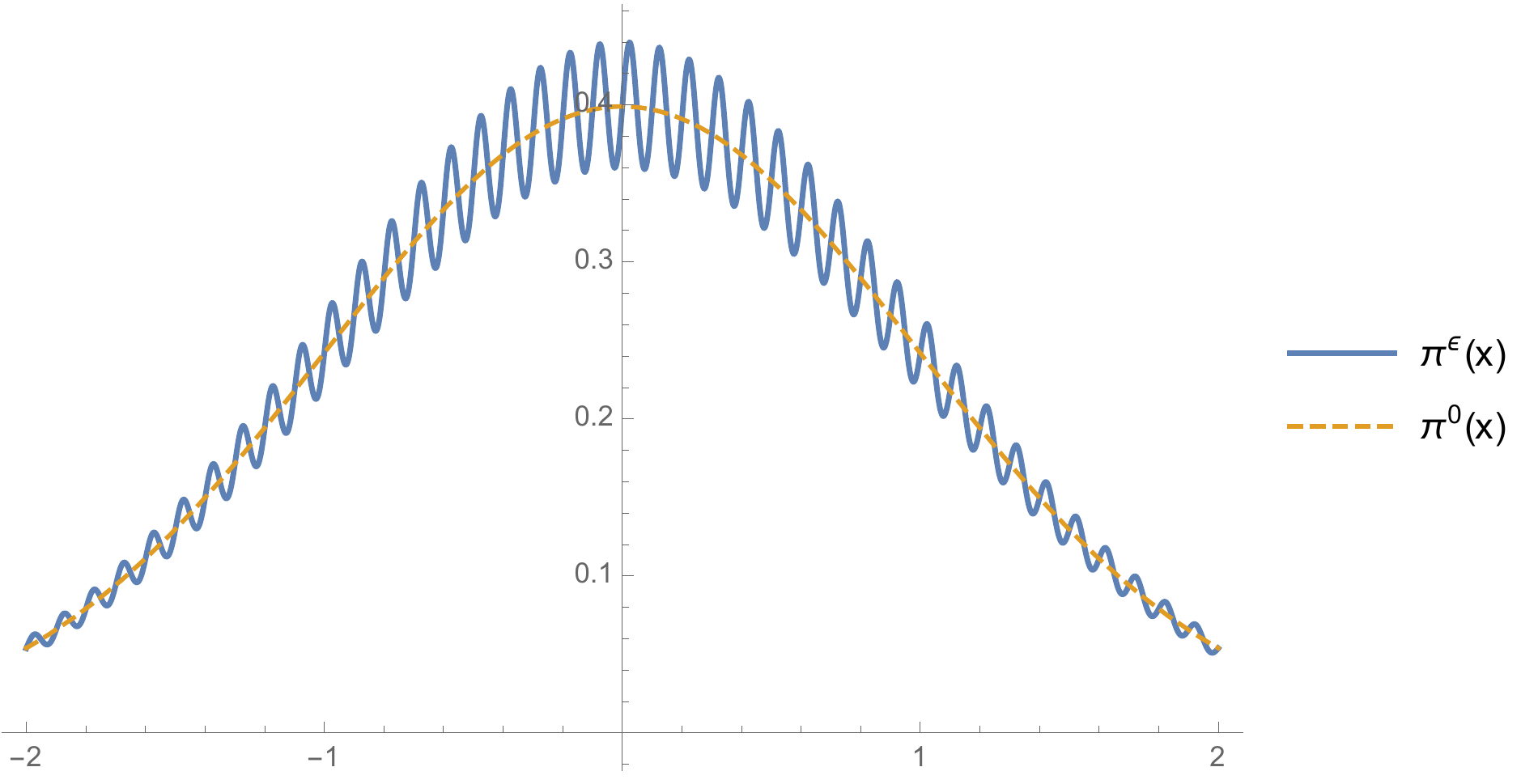}
    \caption{\label{fig:one_dimensional_cex} Plot of $\pi^\epsilon$ and $\pi^0$with $\epsilon = \alpha = 0.1$ and $\sigma=1.0$}
    \end{subfigure}%
    ~ 
    \begin{subfigure}[b]{0.45\textwidth}
        \centering
    \includegraphics[width=\textwidth]{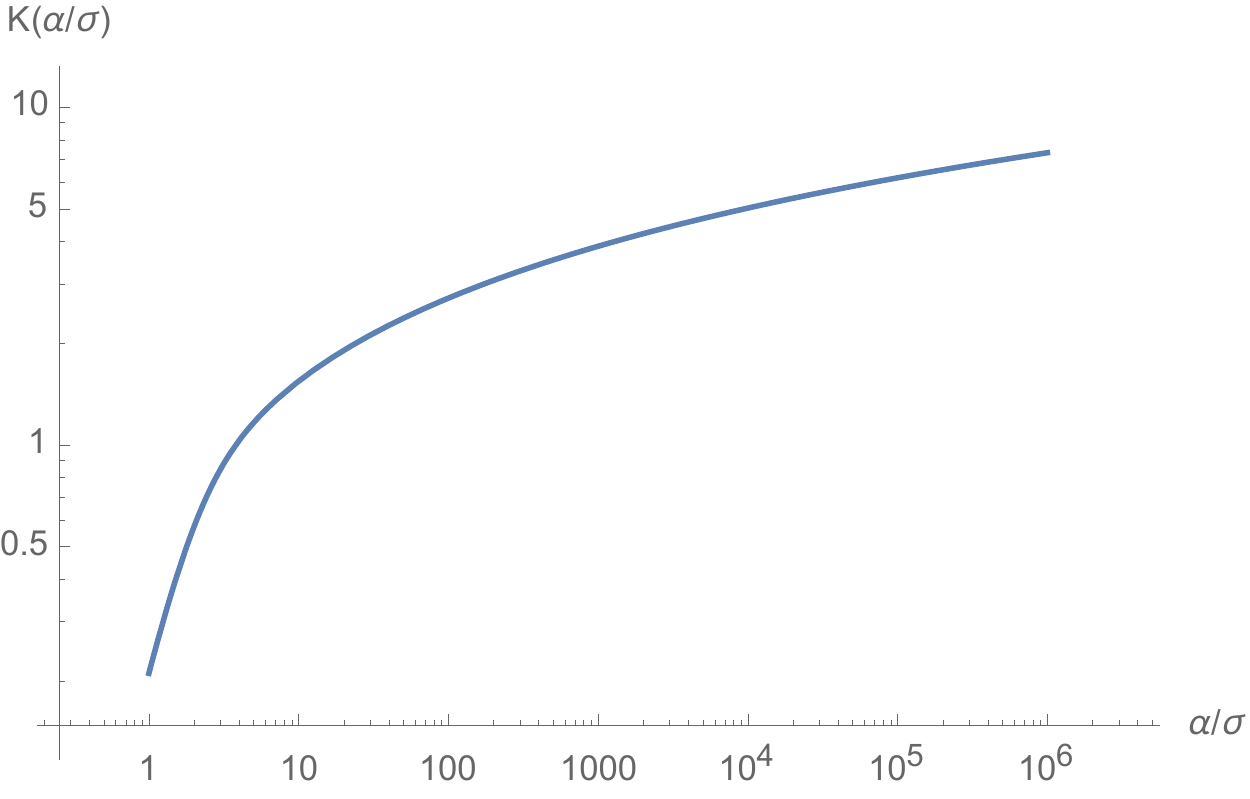}
    \caption{\label{fig:one_dimensional_k}
      Plot of $K(\alpha/\sigma)$ as a function of $\alpha/\sigma$.}
    \end{subfigure}
    \caption{Error between $\pi^\epsilon(x) \propto \exp(-V^\epsilon(x)/\sigma)$ and effective distribution $\pi^0$.}
\end{figure}

\end{example}

\section{Proof of weak convergence}
\label{sec:homog_proof}
In this section we show that over finite time intervals $[0,T]$, the process $X_t^\epsilon$ converges weakly to a process $X^0_t$ which is uniquely identified as the weak solution of a coarse-grained SDE.  The approach we adopt is based on the classical martingale methodology of \cite[Section 3]{bensoussan1978asymptotic}.   The proof of the homogenization result is split into three steps.  
\begin{enumerate}
  \item We construct an appropriate test function which is used to decompose the fluctuations of the process $X_t^\epsilon$ into a martingale part and a term which goes to zero as $\epsilon\rightarrow 0$.  
  \item Using this test function, we demonstrate that the path measure $\mathbb{P}^\epsilon$ corresponding to the family $\left\lbrace \left(X_t^\epsilon\right)_{t\in[0,T]} \right\rbrace_{0 < \epsilon \leq 1}$ is tight in $C([0,T];\mathbb{R}^d)$.  
  \item Finally, we show that any limit point of the family of measures must solve a well-posed martingale problem, and is thus unique.
\end{enumerate}

The test functions will be constructed by solving a recursively defined sequence of Poisson equations on $\mathbb{R}^d$.  We first provide a general well-posedness result for this class of equations.

\begin{proposition}
\label{prop:existence_poisson_general}
  For fixed $(x_0,\ldots, x_{k-1})\in\mathbb{X}_{k-1}$, let $\mathcal{S}_k$ be the operator given by
  \begin{equation}
  \label{eq:poisson_eq_general}
    \mathcal{S}_k = \frac{1}{\rho(x_0, \ldots, x_k)}\nabla_{x_k}\cdot\left(\rho(x_0,\ldots, x_k)D(x_0, \ldots, x_k)\nabla_{x_k} u(x_0,  \ldots, x_{k})\right),\quad f\in C^2(\mathbb{T}^d),
  \end{equation}
   and suppose that $\rho$ is smooth and uniformly positive and bounded, and the tensor $D(x_0,\ldots, x_k)$ is smooth and uniformly positive definite on $\mathbb{X}_{k}$.  Given a function $h$ which is smooth with bounded derivatives, such that for each $(x_0,\ldots, x_{k-1})\in\mathbb{X}_{k-1}$: 
  \begin{equation}
  \label{eq:centering_condition}
    \int h(x_0, \ldots, x_k)\rho(x_0,\ldots, x_k)\,d{x_k} = 0.
  \end{equation}
  Then there exists a unique, mean-zero solution $u \in {H}^1(\mathbb{T}^d)$, to the Poisson equation on $\mathbb{T}^d$ given by
  \begin{equation}
  \label{eq:sk_equation}
   \mathcal{S}_k u(x_0, \ldots, x_k) = h(x_0, \ldots, x_{k}),
  \end{equation}
  which is smooth and bounded with respect to the variable $x_k \in \mathbb{T}^d$ as well as the parameters $x_0, \ldots, x_{k-1} \in \mathbb{X}_{k-1}$.
\end{proposition}
\begin{proof}
\begin{sloppypar}
Since $\rho(\cdot)$ and $D(\cdot)$ are strictly positive, for fixed values of $x_0, \ldots, x_{k-1}$,  the operator $\mathcal{S}_k$ is uniformly elliptic, and since $\mathbb{T}^d$ is compact, $\mathcal{S}_k$ has compact resolvent in $L^2(\mathbb{T}^d)$, see \cite[Ch. 6]{evans1998partial} and \cite[Ch 7]{pavliotis2008multiscale}.  The nullspace of the adjoint $\mathcal{S}^*$  is spanned by a single function $\rho(x_0, \ldots, x_{k-1}, \cdot)$.  By the Fredholm alternative, a necessary and sufficient condition for the existence of $u$ is (\ref{eq:centering_condition}) which is assumed to hold.  Thus, there exists a unique solution $u(x_0,\ldots, x_{k-1},\cdot) \in H^1(\mathbb{T}^d)$ having mean zero with respect to $\rho(x_0, \ldots, x_{k})$.  By elliptic estimates and Poincar\'e's inequality, it follows that  there exists $C > 0$ satisfying
$$
\lVert u(x_0, \ldots, x_{k-1}, \cdot) \rVert_{H^1(\mathbb{T}^d)} \leq C\lVert h(x_0, \ldots, x_{k-1}, \cdot) \rVert_{L^2(\mathbb{T}^d)},
$$
for all $(x_0,\ldots, x_{k-1})\in\mathbb{X}_{k-1}$.  Since the components of $D$ and $\rho$ are smooth with respect to $x_k$, standard interior regularity results \cite{gilbarg2015elliptic} ensure that, for fixed ${x_0, \ldots, x_{k-1}\in \mathbb{X}_{k-1}}$, the function $u(x_0, \ldots, x_{k-1},\cdot)$ is smooth.  To prove the smoothness and boundedness with respect to the other parameters $x_0,\ldots, x_{k-1}$, we can apply an approach either similar to \cite{bensoussan1978asymptotic}, by showing that the finite differences approximation of the derivatives of $u$ with respect to the parameters has a limit, or otherwise, by directly differentiating the transition density of the semigroup associated with the generator $\mathcal{S}_k$ ,  see for example \cite{pardoux2001poisson,veretennikov2011sobolev,pardoux2003poisson} as well as \cite[Sec 8.4]{gilbarg2015elliptic}.
\end{sloppypar}
\end{proof}

\begin{remark}
Suppose that the function $h$ in Proposition \ref{prop:existence_poisson_general} can be expressed as 
$$
    h(x_0,\ldots, x_k) = a(x_0, x_1,\ldots, x_k)\cdot \nabla \phi_0(x_0)
  $$
  where $a$ is smooth with all derivatives bounded.  Then the mean-zero solution of (\ref{eq:sk_equation}) can be written as
  $$
    u(x_0, x_1,\ldots, x_k) = \chi(x_0,x_1,\ldots, x_k)\cdot\nabla\phi_0(x_0),
  $$
  where $\chi$ is the classical mean-zero solution to the following Poisson equation 
  $$\mathcal{S}_k\chi(x_0,\ldots, x_k) = a(x_0, \ldots,x_k), \quad (x_0,\ldots, x_{k}) \in \mathbb{X}_{k}.
  $$  
  In particular, $\chi$ is smooth and bounded over $x_0, \ldots, x_k$, so that for some $C > 0$,
  $$
    \lVert\nabla^{\alpha}u(x_0, \ldots, x_k)\rVert_F \leq C\sum_{k=0}^{\alpha_0}\lVert\nabla^{k+1}\phi_0(x_0)\rVert_F,\quad  \forall x_0,x_1,\ldots, x_k,
  $$ 
  for all multi-indices $\alpha=(\alpha_0, \ldots, \alpha_k)$ on the indices $(0, \ldots, k)$, where $\lVert \cdot \rVert_F$ denotes the Frobenius norm.  A similar decomposition is possible for 
  $$
    g(x_0,\ldots, x_k) = A(x_0, x_1,\ldots, x_k):\nabla^2 \phi_0(x_0),
  $$
  where $\nabla^2$ denotes the Hessian.
\end{remark}

\subsection{Contructing the test functions}
It is clear that we can rewrite (\ref{eq:sde_multiscale}) as 
\begin{equation}
\label{eq:nscale_sde_expanded}
  dX^\epsilon_t = -\sum_{i=0}^N \epsilon^{-i}\nabla_{x_i} V(x_1, \ldots,x_N)\Big|_{x_j = X_t^\epsilon/\epsilon^j}\,dt + \sqrt{2\sigma}\,dW_t.
\end{equation}
The generator of $X_t^\epsilon$ denoted by $\mathcal{L}^\epsilon$ can be decomposed into powers of $\epsilon$ as follows
$$
  \mathcal{L}^\epsilon = -\sum_{n=0}^{N}\epsilon^{-n}\nabla_{x_n}V\cdot\nabla_x + \sigma\Delta_{x}.
$$
For functions of the form $f^\epsilon(x) = f(x, x/\epsilon, \ldots, x/\epsilon^N)$ we  write 
$$
  \mathcal{L}^\epsilon f^\epsilon(x) = \sum_{n=0}^{2N}\epsilon^{-n} \mathcal{L}_n f(x_0, x_1 \ldots,x_N)\big|_{x_i = x/\epsilon^i},
$$
where 
$$
  \mathcal{L}_n = e^{ V/\sigma}\sum_{\substack{i, j \in \lbrace 1,\ldots N\rbrace \\ \, i+j = n}}\nabla_{x_i}\cdot\left(\sigma e^{-V/\sigma}\nabla_{x_j}\cdot \right)
$$
Given $\phi_0$, our objective is to construct a test function $\phi^\epsilon$ such that
\begin{align*}
  \phi^\epsilon(x) = &\phi_0(x) + \epsilon\phi_1(x,x/\epsilon) + \ldots +\epsilon^N\phi_N(x, x/\epsilon, \ldots, x/\epsilon^N) \\
  & + \epsilon^{N+1}\phi_{N+1}(x, x/\epsilon, \ldots, x/\epsilon^N) + \ldots + \epsilon^{2N}\phi_{2N}(x, x/\epsilon, \ldots, x/\epsilon^N),
\end{align*}
where $\phi_1, \ldots, \phi_{2N}$ satisfy
\begin{equation}
\label{eq:Lphi_eqn}
  \mathcal{L}^{\epsilon}\phi^\epsilon(x) = F(x) + O(\epsilon),
\end{equation}
for some $F$ which is independent of $\epsilon$.  This is equivalent to the following sequence of $N+1$ equations. 
\begin{subequations}
\begin{align}
\label{eq:asymp_exp1}&\mathcal{L}_{2N} \phi_N + \mathcal{L}_{2N-1} \phi_{N-1} + \ldots \mathcal{L}_N \phi_0 = 0,\\
\label{eq:asymp_exp2}&\mathcal{L}_{2N} \phi_{N+1} + \mathcal{L}_{2N-1}\phi_{N} + \ldots \mathcal{L}_{N-1}\phi_0 = 0,\\
&\notag\vdots\\
&\label{eq:asymp_expN_m_1}\mathcal{L}_{2N} \phi_{2N-1} + \ldots + \mathcal{L}_{1}\phi_0  = 0,\\
& \label{eq:asymp_O1}\mathcal{L}_{2N}\phi_{2N} + \ldots + \mathcal{L}_{0}\phi_0  = F(x),
\end{align}
\end{subequations}
where $F(x)$ is a function of $x$ only.  This generalizes the analogous expansion found in \cite[III-11.3]{bensoussan1978asymptotic}, written for three scales.  These $N+1$ equations correspond to the different powers of $\epsilon$ in an expansion of  $\mathcal{L}^\epsilon \phi^\epsilon$, from $O(\epsilon^{-N})$ to $O(1)$.  For $k=1,\ldots, N$, we note that each term in (\ref{eq:asymp_exp1}), (\ref{eq:asymp_exp2}) to (\ref{eq:asymp_expN_m_1}) has the form
$$
  \sigma e^{V(x_0, \ldots, x_N)/\sigma}\nabla_{x_s}\cdot\left(e^{- V(x_0, \ldots, x_N)/\sigma}\nabla_{t}\phi_{r}\right),
$$
where $k = s + t - r \in \lbrace 1,\ldots, N\rbrace$.  Suppose that $s = 0$, so that $t = k +r$, where $t \in \lbrace 1, \ldots, N\rbrace$ and $r \in \lbrace 0, \ldots, N-1 \rbrace$.  Thus $r < t$, which is a contradiction.  It follows  necessarily that $s \geq 1$, for every term in the first $N$ equations.  In particular, since we have
$$
  V(x_0, \ldots, x_{N}) = V_0(x_0) + V_1(x_0,\ldots, x_N),  
$$ 
we can rewrite the first $N$ equations as 
\begin{subequations}
\label{eq:factorized}
\begin{align}
\label{eq:asymp_exp1_2}&\mathcal{A}_{2N} \phi_N + \mathcal{A}_{2N-1} \phi_{N-1} + \ldots \mathcal{A}_N \phi_0 = 0,\\
\label{eq:asymp_exp2_2}&\mathcal{A}_{2N} \phi_{N+1} + \mathcal{A}_{2N-1}\phi_{N} + \ldots \mathcal{A}_{N-1}\phi_0 = 0,\\
&\notag\vdots\\
&\label{eq:asymp_expN_m_1_2}\mathcal{A}_{2N} \phi_{2N-1} + \ldots + \mathcal{A}_{1}\phi_0  = 0,
\end{align}
\end{subequations}
where
$$
\mathcal{A}_n f = \sigma e^{ V_1(x_0, \ldots, x_N)/\sigma}\sum_{\substack{i\in\lbrace 1,\ldots, N\rbrace \\ j\in\lbrace 0,\ldots, N-1\rbrace \\  i+j=n}} \nabla_{x_i}\cdot\left(e^{- V_1(x_0, \ldots, x_N)/\sigma}\nabla_{x_j}f \right)
$$

\begin{sloppypar}
Before constructing the test functions, we first we introduce the sequence of spaces on which the sequence of correctors will be constructed.  Define $\mathcal{H}$ to be the space of functions on the extended state space, i.e. $\mathcal{H} = L^2(\mathbb{X}_k)$, where $\mathbb{X}_k$ is defined by (\ref{eq:XK}).  We construct the following sequence of subspaces of $\mathcal{H}$.  Let
$$
  \mathcal{H}_{N} = \left\lbrace f \in \mathcal{H} \, : \, \int f(x_0, \ldots, x_{N})e^{- V_1/\sigma}\,dx_{N} = 0 \right\rbrace,
$$
Then clearly $\mathcal{H} = \mathcal{H}_{N} \oplus \mathcal{H}_{N}^\perp$.  Suppose we have defined $\mathcal{H}_{N-k+1}$ then we can define $\mathcal{H}_{N-k}$ inductively by
$$
  \mathcal{H}_{N-k} =  \left\lbrace f \in \mathcal{H}_{N-k+1} \, : \, \int f(x_0, \ldots, x_{N-k})Z_{N-k}(x_0, \ldots, x_{N-k})\,dx_{N-k} = 0 \right\rbrace,
$$
where $Z_{i}(x_0, \ldots, x_i) = \int\ldots\int e^{- V_1(x_0, \ldots, x_N)/\sigma}\,d{x_{i+1}}\,dx_{i+2}\ldots \,dx_N = 0.$   Clearly, we have that $\mathcal{H}_1 \oplus \mathcal{H}_1^\bot \oplus  \ldots \oplus \mathcal{H}_N^\bot = \mathcal{H}$.  We now construct a series of correctors $\theta_1, \ldots, \theta_N$ which are used to define the test functions.  Define 
 $$\mathcal{K}_N(x_0, x_1, \ldots, x_N) = e^{-V_1(x_0, x_1, \ldots,x_N)/\sigma^2}I.$$
 We note that the matrix $\mathcal{K}_N$ is uniformly positive definite over $\mathbb{X}_N$.  Fixing ${x_0, x_1, \ldots, x_{N-1}}$, let $\theta_N$ be the solution of the vector--valued Poisson equation
\begin{equation}
\label{eq:poisson_eqn_theta_n}
  \nabla_{x_N}\cdot\left(\mathcal{K}_N\left(\nabla_{x_N}\theta_N + I\right)\right) = 0, \quad x_N \in \mathbb{T}^d,
\end{equation}
where the notation $\left(\nabla_{x}\theta\right)_{i, j} = \partial_{x_{j}} \theta_{i},\quad i, j \in \lbrace 1, \ldots, d\rbrace$ is used.  By Proposition \ref{prop:existence_poisson_general}, for each $(x_0, \ldots, x_{N-1})$ there exists a unique smooth solution $\theta_N(x_0,\ldots, x_{N-1}, \cdot)$ which is also smooth with respect to the parameters $x_0, \ldots, x_{N-1}$.  Now, suppose that $\mathcal{K}_i$ and $\theta_i$ have been defined for $i \in \lbrace N, \ldots, N-k+1 \rbrace$, define
\begin{equation}
\label{eq:K_i}
\begin{aligned}
  \mathcal{K}_{N-k}&(x_0, x_1, \ldots, x_{N-k}) \\
  &= \int (I + \nabla_N\theta_N^\top)\cdots (I + \nabla_{N-k+1}\theta_{N-k+1}^\top)e^{-V_1/\sigma^2} \,dx_{N}\ldots dx_{N-k+1}.
  \end{aligned}
\end{equation}
Then by Lemma \ref{lem:corr_existence_ellipticity} the matrix $\mathcal{K}_{N-k}$ is strictly positive definite over $(x_0, \ldots, x_{N-k})$ and so there exists a unique vector--valued solution $\theta_{N-k}$ in $\left(H^1(\mathbb{T}^d)\cap \mathcal{H}_{N-k}\right)^{d}$ to the Poisson equation:
\begin{equation}
  \label{eq:theta_k}
  \nabla_{x_{N-k}}\cdot\left(\mathcal{K}_{N-k}(\nabla_{x_{N-k}}\theta_{N-k} + I)\right) = 0, \quad x_{N-k} \in \mathbb{T}^d.
\end{equation}
\end{sloppypar}

\begin{proposition} 
\label{prop:test_functions}
Given $\phi_0 \in C^\infty(\mathbb{R}^d)$, there exist smooth functions $\phi_i$ for $i=1,\ldots, 2N-1$ such that equations  (\ref{eq:asymp_exp1_2})-(\ref{eq:asymp_expN_m_1_2}) are satisfied, and moreover we have the following pointwise estimates, which hold uniformly on $x_0,\ldots, x_k \in \mathbb{X}_k$:
\begin{equation}
\label{eq:pointwise_estimates}
  \lVert \nabla^{\alpha} \phi_i(x_0,\ldots, x_k)\rVert_F\leq C\sum_{l=1}^{\alpha_0+2}\lVert\nabla^{l}_{x_0}\phi_0(x_0)\rVert_F,
\end{equation}
for some constant $C > 0$, and all multiindices $\alpha$ on $(0, \ldots, k)$, and all $0 \leq k \leq i \leq 2N-1$.   Finally, equation  (\ref{eq:asymp_O1})  is satisfied with 
\begin{equation}
\label{eq:F}
F(x) = \frac{1}{Z(x)}\nabla_{x_0}\left(\mathcal{K}_1(x)\nabla_{x_0}\phi_0(x)\right).
\end{equation}
\end{proposition}
\begin{proof}
We start from the $O(\epsilon^{-N})$ equation.  Since the operator $\mathcal{A}_{2N}$ has a compact resolvent in $L^2(\mathbb{T}^d)$,  by the Fredholm alternative a necessary and sufficient condition for $\phi_N$ in (\ref{eq:asymp_exp1}) to have a solution is that
 $$
  \int \left(\mathcal{A}_{2N-1} \phi_{N-1} + \mathcal{A}_{2N-2}\phi_{N-2} + \ldots + \mathcal{A}_{N} \phi_0\right) e^{-V/\sigma}\,dx_N = 0.
 $$
We can check that the only non-zero terms in the above summation are:
$$
  \mathcal{A}_i \phi_i = \sigma e^{V/\sigma}\nabla_{x_N}\cdot\left(e^{-V/\sigma}\nabla_{x_{N-i}} \phi_{i}\right),
$$
for $i = 1,\ldots, N$, so that the compatibility condition holds, by the periodicity of the domain.  Then $\theta_N$ defined by (\ref{eq:poisson_eqn_theta_n}) is the unique mean--zero solution of 
$$
  \mathcal{A}_{2N} \theta_N = \nabla_{x_N}\cdot e^{-V/\sigma},
$$
then the solution $\phi_N$ to (\ref{eq:asymp_exp1}) can be written as 
\begin{equation}
  \phi_N = \theta_N \cdot \left(\nabla_{x_{N-1}}\phi_{N-1} + \ldots + \nabla_{x_{0}}\phi_{0} \right) + r_N^{(1)}(x_0, \ldots, x_{N-1}),
\end{equation}
where  
$$\theta_N\cdot(\nabla_{x_N-1}\phi_{N-1} + \ldots + \nabla_{x_0}\phi_0) \in  \mathcal{H}_N$$
and $r_N^{(1)} \in \mathcal{H}_N^\bot$ has not yet been specified.  A sufficient condition for $\phi_{N+1}$ to have a solution in (\ref{eq:asymp_exp2}) is that
\begin{equation}
\label{eq:compat_cndn2}
\int \left(\mathcal{A}_{2N-1}\phi_N + \ldots + \mathcal{A}_{N-2} \phi_1 + \mathcal{A}_{N-1}\phi_0 \right) e^{-V/\sigma}\,dx_N = 0.
\end{equation}
Since $r_N^{(1)}$ does not depend on $x_N$ it follows that:
$$
  \int e^{-V/\sigma}\mathcal{A}_{2N-1} \phi_N \,d{x_N} = \nabla_{x_{N-1}}\cdot\left(\int e^{-V/\sigma}\nabla_{x_{N}}\theta_N \left(\nabla_{x_{N-1}}\phi_{N-1} \ldots + \nabla_{x_0}\phi_0\right)\,dx_N\right),
$$
thus (\ref{eq:compat_cndn2}) can be written as
\begin{align*}
 0= &\nabla_{x_{N-1}}\cdot\left(\int e^{-V/\sigma}\nabla_{x_{N}}\theta_N \left(\nabla_{x_{N-1}}\phi_{N-1} \ldots + \nabla_{x_0}\phi_0\right)\,dx_N\right)\\
  & + \nabla_{x_{N-1}}\cdot\left(\int e^{-V/\sigma} \,dx_{N} \left(\nabla_{x_{N-1}}\phi_{N-1} + \ldots +\nabla_{x_0}\phi_0\right)\right),
\end{align*}
resulting in the following equation for $\phi_{N-1}$:
\begin{equation}
\label{eq:corrector_o2}
 \nabla_{x_{N-1}}\cdot\left(\mathcal{K}_N\nabla_{x_{N-1}}\phi_{N-1}\right) = -\nabla_{x_{N-1}}\cdot\mathcal{K_N}\left(\nabla_{x_{N-2}}\phi_{N-2} + \ldots + \nabla_{x_0}\phi_0\right) = 0,
\end{equation}
where 
$$
\mathcal{K}_N = \int\left(I + \nabla_{x_{N}}\theta_N \right) e^{-V/\sigma}\, dx_N.
$$
By Lemma \ref{lem:corr_existence_ellipticity}, for fixed $x_0, x_1, \ldots, x_{N-1}$ the tensor $\mathcal{K}_N$ is uniformly positive definite over $x_{N-1} \in \mathbb{T}^d$. As a consequence, the operator defined in (\ref{eq:corrector_o2}) is uniformly elliptic, with adjoint nullspace spanned by $Z_N(x_0, x_1,\ldots, x_{N-1})$.  Since the right hand side has mean zero, this implies that a solution $\phi_{N-1}$ exists.  Indeed,  we can write $\phi_{N-1}$ as 
$$
  \phi_{N-1} = \theta_{N-1}\cdot \left(\nabla_{x_{N-2}}\phi_{N-2} + \ldots + \nabla_{x_0}\phi_0\right) + r_{N-1}^{(1)}(x_0, \ldots, x_{N-2}),
$$
where $r^{(1)}_{N-1} \in \mathcal{H}_{N-1}^\bot$ is still unspecified.  Since (\ref{eq:compat_cndn2}) has been satisfied, it follows from Proposition \ref{prop:existence_poisson_general} that there exists a unique decomposition of $\phi_{N+1}$ into
$$
\phi_{N+1}(x_0, x_1, \ldots, x_{N}) = \widetilde{\phi}_{N+1}(x_0, x_1, \ldots, x_{N}) + r_{N+1}^{1}(x_0, x_1, \ldots, x_{N-1}),
$$
where  $\widetilde{\phi}_{N+1} \in \mathcal{H}_{N}$ and $r_{N+1}^{(1)} \in \mathcal{H}_{N}^\bot,$  such that $r_{N+1}^{(1)}$ is still unspecified.  For the sake of illustration we now consider the $O(\epsilon^{-(N-2)})$ equation in (\ref{eq:factorized}).  This equation for $\phi_{N+2}$ has a solution if and only if
$$
  \int \left(\mathcal{A}_{2N-1}\phi_{N+1} +  \mathcal{A}_{2N-2}\phi_{N} + \ldots + \mathcal{A}_{N-2}\phi_{0}\right)\, e^{-V/\sigma}dx_N = 0.
$$
Fixing the variables $x_0, \ldots, x_{N-2}$, we can rewrite the above equation as:
\begin{equation}
\label{eq:rN1}
  \widetilde{\mathcal{A}}_{2N-2} r_N^{(1)} := \nabla_{N-1}\cdot \left(Z_{N-1} \nabla_{N-1}r_{N}^{(1)}\right) =  -RHS,
\end{equation}
where the $RHS$ contains all the remaining terms.  We note that all the functions of $x_{N-1}$ in the RHS are known,  so that all the remaining undetermined terms can be viewed as constants for fixed $x_0, \ldots, x_{N-2} \in \mathbb{X}_{N-2}$.  A necessary and sufficient condition for a unique mean zero solution to exist to (\ref{eq:rN1}) is that the RHS has integral zero with respect to $x_{N-1}$, which is equivalent to:
$$
  \nabla_{N-2}\cdot\left(\int\int \left(\nabla_N \phi_N + \nabla_{N-1}\phi_{N-1} + \ldots + \nabla_{0}\phi_0 \right) e^{-V/\sigma}\,dx_N dx_{N-1}\right)=0,
$$
or equivalently:
$$
  \nabla_{N-2}\cdot\left(\mathcal{K}_{N-2}\left(\nabla_{N-2}\phi_{N-2} + \ldots + \nabla_0 \phi_0\right)\right)=0.
$$
Once again, this implies that
$$
  \phi_{N-2} = \theta_{N-2}\cdot\left(\nabla_{N-3}\phi_{N-3} + \ldots + \nabla_{0}\phi_0\right) + r_{N-2}^{(1)}(x_0, \ldots, x_{N-3}),
$$
where $r_{N-2}^{(1)} \in \mathcal{H}_{N-2}^{\bot}$ is  unspecified.  Since the compatibility condition holds, by Proposition \ref{prop:existence_poisson_general} equation (\ref{eq:rN1}) has a solution, so that we can write
$$
  r_{N}^{(1)}(x_0, \ldots, x_{N-1}) = \widetilde{r}_{N}^{(1)}(x_0, \ldots, x_{N-1}) + r_{N}^{(2)}(x_0, \ldots, x_{N-2}),
$$
where $\widetilde{r}^{(1)}_{N} \in \mathcal{H}_{N-1}$ is the unique smooth solution of (\ref{eq:rN1}) and for some $r_{N}^{(2)} \in \mathcal{H}_{N-1}^\bot$.
\\\\
For the inductive step, suppose that for some $k < N$, the functions $\phi_N, \ldots \phi_{N\pm(k-1)}$ have all been determined.  We shall consider the case when $k$ is even, noting that the $k$ odd case follows \emph{mutatis mutandis}.  From the previous steps, each term in 
$$
\phi_{N+k-2}, \phi_{N+k-4}, \ldots, \phi_{N-k-2},
$$
admits a decomposition such that in each case we can write:
$$
  \phi_{N+k - 2i} = \widetilde{\phi}_{N+k-2i} + r^{(k/2-i)}_{N+k-2i},
$$
where 
$$\widetilde{\phi}_{N+k-2i}  \in \mathcal{H}_{k/2-i},$$ 
has been uniquely specified, and the remainder term
$$r_{N+k-2i}^{(k/2-i)} \in \mathcal{H}_{k/2-i}^\bot,$$ 
remains to be determined.  The $O(\epsilon^{N-k})$ equation is given by
\begin{equation}
\label{eq:ON_minus_k}
  \mathcal{A}_{2N}\phi_{N+k} + \mathcal{A}_{2N-1}\phi_{N+k-1} + \ldots + \mathcal{A}_{N-k}\phi_{0} = 0. 
\end{equation}
Following the example of the $O(\epsilon^{N-2})$ step.  In descending order we successively apply the compatibility conditions which must be satisfied for the equations involving $r_{N+k}^{(1)}, \ldots, r_{N-k-2}^{(k-1)}$ of the form
\begin{equation}
\label{eq:general_remainder}
  \mathcal{\widetilde{A}}_{2N-2k-2i}r_{N+k-2i}^{(k/2-i)} = RHS,
\end{equation}
where in (\ref{eq:general_remainder}), all terms dependent on the variable $x_{k/2-i}$ have been specified uniquely and where
$$
  \widetilde{A}_{2N-2k-2i} u = \nabla_{x_{N-k-i}}\cdot\left(Z_{N-k-i}\nabla_{x_{N-k-i}}u\right).
$$  
This results in (\ref{eq:ON_minus_k}) being integrated with respect to the variables $N, \ldots, N-k+1$.  In particular, all terms $\mathcal{A}_{2N-j}\phi$ for $j=0, \ldots, k-1$ will have integral zero, and thus vanish.  The resulting equation is then
\begin{equation}
\label{eq:inductive_step_integral}
  \int\ldots\int \left(\mathcal{A}_{2N-k}\phi_{N} + \ldots + \mathcal{A}_{N-k}\phi_0\right)e^{-V/\sigma}\,dx_N\ldots dx_{N-k+1} = 0.
\end{equation}
Moreover, since the function $\phi_{N-i}$ depends only on the variables $x_0, \ldots, x_{N-i}$ , then (\ref{eq:inductive_step_integral}) must be of the form
$$
  \nabla_{N-k}\cdot\left(\int\ldots\int \left(\nabla_{x_N}\phi_N + \ldots \nabla_{x_{N-1}}\phi_{N-1} + \ldots \nabla_{x_0}\phi_0\right) e^{-V/\sigma}\,dx_{N}\ldots dx_{N-k+1}\right) = 0.
$$
We now apply the inductive hypothesis to see that
\begin{align*}
& \int \left(\nabla_{x_N} \phi_N + \ldots \nabla_0\phi_0\right) e^{-V/\sigma}\,dx_{N}\cdots dx_{N-k+1} \\ &= \int \int\left(\nabla_{x_N} \theta_N + I\right)dx_N\left(\nabla_{N-1}\phi_{N-1} + \ldots + \nabla_{x_0}\phi_0\right) e^{-V/\sigma}\,dx_{N-1}\cdots dx_{N-k+1}\\
& = \int \int\int\left(\nabla_{x_N} \theta_N + I\right)\,dx_{N}\left(\nabla_{x_{N-1}} \theta_{N-1} + I\right)\,dx_{N-1}\left(\nabla_{x_{N-2}}\phi_{N-2} + \ldots + \nabla_{x_0}\phi_0\right) e^{-V/\sigma}\,dx_{N-2}\cdots dx_{N-k+1}\\
&\vdots\\
&= \mathcal{K}_{N-k+1}\left(\nabla_{x_{N-k}} \phi_{N-k} + \ldots \nabla_{x_0}\phi_0\right).
\end{align*}
Thus, the compatibility condition for the $O(\epsilon^{N-k})$ equation reduces to the elliptic PDE
$$
  \nabla_{x_k}\cdot\left(\mathcal{K}_{x_{N-k}}\left(\nabla_{N-k} \phi_{x_{N-k}} + \ldots \nabla_{x_0}\phi_0\right)\right) = 0,
$$
so that $\phi_{N-k}$ can be written as 
$$
  \phi_{N-k} = \theta_{N-k}\left(\nabla_{x_{N-k-1}}\phi_{x_{N-k-1}} + \ldots \nabla_{x_{0}}\phi_0\right) + r_{N-k}^{(1)},
$$
where $r_{N-k}^{(1)}$ is an element of $\mathcal{H}_{N-k}^\bot,$ which is yet to be determined.  Moreover, each remainder term $r_{N+k-2i}^{(k/2-i)}$ can be further decomposed as 
$$
r_{N+k-2i}^{(k/2-i)} = \widetilde{r}_{N+k-2i}^{(k/2-i)}  + {r}_{N+k-2i}^{(k/2-i+1)}, 
$$
where 
$$\widetilde{r}_{N+k-2i}^{(k/2-i)}  \in \mathcal{H}_{k/2-i+1},$$
is uniquely determined and 
$${r}_{N+k-2i}^{(k/2-i+1)}  \in \mathcal{H}_{k/2-i+1}^\bot,$$
is still unspecified.  Continuing the above procedure inductively, starting from a smooth function $\phi_0$ we construct a series of correctors $\phi_1, \ldots, \phi_{2N-1}$.  
\\\\
We now consider the final equation (\ref{eq:asymp_O1}).  Arguing as before, we note that we can rewrite (\ref{eq:O1}) as 
\begin{equation}
\mathcal{A}_{2N}\phi_{2N} + \ldots \mathcal{A}_{N+1}\phi_{N+1} = F(x) - \sum_{i=1}^{N}\mathcal{L}_i \phi_i.
\end{equation}
A necessary and sufficient condition for $\phi_{2N}$ to have a solution is that
\begin{equation}
\label{eq:O1}
\begin{aligned}
  \int_{\mathbb{T}^d} &\left(\mathcal{A}_{2N-1}\phi_{2N-1} + \ldots + \mathcal{A}_{N+1} \phi_{N+1} \right)e^{-V/\sigma}\,dx_N \\
   & = \int_{\mathbb{T}^d} \left(F(x)-\sum_{i=1}^N\mathcal{L}_i\phi_i \right)e^{-V/\sigma}\,dx_N.
  \end{aligned}
\end{equation}
At this point, the remainder terms will be of the form 
$$
r_{2N-2}^{(1)}, r_{2N-4}^{(2)}, \ldots r_{2N-2k}^{(k)}, \ldots, r_{2}^{(1)},
$$
such that $r_{2N-2i}^{(i)} \in \mathcal{H}_{i}^{\bot}$, is unspecified.  Starting from $r_{2N-2}^{(1)}$ a necessary and sufficient condition for the remainder $r_{2N-2i}^{(i)}$ to exist is that the integral of the equation with respect to $dx_{N-i}$ vanishes, i.e.
\begin{equation}
\label{eq:O1_final}
\begin{aligned}
 F(x)Z(x) &= \int_{(\mathbb{T}^d)^N} \left(\mathcal{A}_{2N-1}\phi_{2N-1} + \ldots \mathcal{A}_{N+1}\phi_{N+1}\right)e^{-V/\sigma}\,dx_N dx_{N-1}\ldots dx_{1} \\
          &\quad +  \int_{(\mathbb{T}^d)^N} \left(\mathcal{L}_{N}\phi_{N} + \ldots \mathcal{L}_{1}\phi_1\right)e^{-V/\sigma}\,dx_N dx_{N-1}\ldots dx_{1}
\end{aligned}
\end{equation}
where $$Z(x) = \int_{\mathbb{T}^d}\ldots\int_{\mathbb{T}^d} e^{-V/\sigma}\,dx_N\ldots dx_{1}.$$  As above, after simplification, (\ref{eq:O1_final}) becomes
$$
  \nabla_{x_0}\cdot\left(\nabla_{x_N}\phi_N + \ldots +\nabla_{x_0}\phi_0 \right) =Z(x)F(x), 
$$
which can be written as 
$$
  \frac{\sigma}{Z(x)}\nabla_{x_0}\cdot\left(\int_{(\mathbb{T}^d)^N}\left(I + \nabla_{x_N}\theta_N\right)\cdot\ldots \cdot\left(I + \nabla_{x_1}\theta_1\right)e^{-V/\sigma}\,dx_N\ldots dx_1\nabla_{x_0}\phi_0\right) = F(x),
$$
or more compactly
$$
  F(x) = \frac{\sigma}{Z(x)}\nabla_{x_0}\cdot\left(\mathcal{K}_1(x) \nabla_{x_0}\phi_0(x)\right),
$$
where the terms in the right hand side have been specified and are unique.   Thus, the $O(1)$ equation (\ref{eq:O1_final}) provides a unique expression for $F(x)$.  Moreover, for each $i=1,\ldots, N-1$, there exists a smooth unique solution $r_{2N-2i}^{(i)} \in  \mathcal{H}_{i-1}$ and $\phi_{2N} \in \mathcal{H}_{2N}$ by Proposition \ref{prop:existence_poisson_general}.  
\\\\
Note that we have not uniquely identified the functions $\phi_1, \ldots, \phi_{2N}$, since after the above $N$ steps there will be remainder terms which are still unspecified.  However, conditions (\ref{eq:asymp_exp1_2})-(\ref{eq:asymp_expN_m_1_2}) will hold for any choice of remainder terms which are still unspecified.  In particular, we can set all the remaining unspecified remainder terms to $0$.  Moreover,  every Poisson equation we have solved in the above steps has been of the form:
$$
\mathcal{S}_k u(x_0,\ldots, x_k) = a(x_0,\ldots,x_k)\cdot\nabla_{x_0}\phi_0(x_0) + A(x_0,\ldots, x_k):\nabla^2_{x_0}\phi_0(x_0),
$$
where $\mathcal{S}_k$ is of the form (\ref{eq:poisson_eq_general}), and $a$ and $A$ are uniformly bounded with bounded derivatives. In particular, from the remark following Proposition \ref{prop:existence_poisson_general} the pointwise estimates (\ref{eq:pointwise_estimates}) hold. 
\end{proof}

\begin{remark}
Although we do not have an explicit formula for the test functions, for $i = 1,\ldots, N$, we have that an expression for the gradient of $\phi_i$ in terms of the correctors $\theta_i$:
$$
\nabla_{x_i}\phi_i = \nabla_{x_i}\theta_i(1 + \nabla_{x_{i-1}}\theta_{x_{i-1}})\cdot\cdots\cdot (1 + \nabla_{x_{1}}\theta_{x_{1}})\nabla_{x_0}\phi_0.
$$
As we shall see, these are the only terms that are required for the calculation of the homogenized diffusion tensor, thus we can obtain an explicit characterisation of the effective coefficients.
\end{remark}

\subsection{Tightness of Measures}
In this section we establish the weak compactness of the family of measures corresponding to $\lbrace X_t^\epsilon : 0 \leq t \leq T\rbrace_{0 < \epsilon \leq 1\rbrace}$ in $C([0,T]; \mathbb{R}^d)$ by establishing tightness. Following \cite{pardoux2001poisson}, we verify the following two conditions which are a slight modification of the sufficient conditions stated in \cite[Theorem 8.3]{billingsley2008probability}. 

\begin{lemma}
The collection $\lbrace X_t^\epsilon\, : \, 0 \leq t \leq T\rbrace_{\lbrace 0 < \epsilon \leq 1\rbrace}$ is relatively compact in $C([0,T]; \mathbb{R}^d)$ if it satisfies: 
\begin{enumerate}
  \item \label{it:tight1}For all $\delta > 0$, there exists $M > 0$ such that 
  $$
    \mathbb{P}\left(\sup_{0 \leq t \leq T}|X_t^\epsilon| > M \right) \leq \delta,   \quad 0 < \epsilon \leq 1.
  $$
  \item \label{it:tight2} For any $\delta > 0$, $M > 0$, there exists $\epsilon_0$ and $\gamma$ such that
  $$
  \gamma^{-1}\sup_{0<\epsilon<\epsilon_0}\sup_{0\leq t_0 \leq T}\mathbb{P}\left(\sup_{t\in[t_0, t_0+\gamma]}\left|X_t^\epsilon - X_{t_0}^\epsilon\right| \geq \delta \, ; \, \sup_{0 \leq s \leq T}|X_s^\epsilon| \leq M\right)\leq \delta.
  $$
  % $$
  %   \sup_{0 < \epsilon \leq \epsilon_0}\mathbb{P}\left(\sup_{0 \leq t_0 \leq T}\sup_{t\in[t_0, t_0+\gamma]}|X_t^\epsilon - X_{t_0}^\epsilon| \geq \delta \, ; \sup_{0 \leq s \leq t}|X_t^\epsilon| \leq M\right) \leq \delta.
  % $$
\end{enumerate}
\end{lemma}

To verify condition \ref{it:tight1} we follow the approach of \cite{pardoux2001poisson} and consider a test function of the form $\phi_0(x) = \log(1 + |x|^2)$.  The motivation for this choice is that while $\phi_0(x)$ is increasing, we have that
\begin{equation}
\label{eq:phi0bounds}
\sum_{k=1}^{3}(1 + |x|)^l\lVert \nabla_x^l \phi_0(x)\rVert_F \leq C.
\end{equation}
Let $\phi_1, \ldots, \phi_{2N-1}$ be the first $2N-1$ test functions constructed in Proposition \ref{prop:test_functions}.  Consider the test function
\begin{equation}
\label{eq:tightness_test_function}
  \phi^\epsilon(x) = \phi_0(x) + \epsilon\phi_1(x, x/\epsilon) + \ldots + \epsilon^{2N-1}\phi_{2N-1}(x, x/\epsilon, \ldots, x/\epsilon^{2N-2}, x/\epsilon^{2N-1}).
\end{equation}
Applying It\^{o}'s formula, we have that
$$
\phi^\epsilon(X^\epsilon_t) = \phi^\epsilon(x) + \int_0^t G(X_s^\epsilon)\,ds + \sum_{i=0}^{N}\sum_{k=0}^{2N-1} \epsilon_i  \int_0^t\nabla_{x_i}\phi_j\,dW_s,
$$
where $G(x)$ is a smooth function consisting of terms of the form:
\begin{equation}
\label{eq:tightness_decomposition}
\epsilon^{i+j-k}e^{V/\sigma}\nabla_{x_i}\cdot\left(e^{- V/\sigma} \nabla_{j}\phi_{k}\right)(x, x/\epsilon, \ldots, x/\epsilon^N),
\end{equation}
To obtain relative compactness we need to individually control the terms arising in the drift.  More specifically, we must show that
\begin{equation}
\label{eq:term1}
  \mathbb{E} \sup_{0\leq t \leq T} \int_0^t\left|e^{V/\sigma}\nabla_{x_i}\cdot\left(e^{- V/\sigma} \nabla_{j}\phi_{k}\right)(X_s^\epsilon, X_s^\epsilon/\epsilon, \ldots, X_s^\epsilon/\epsilon^N)\,ds\right| < \infty,
\end{equation}
where $i+j-k \geq 0$, and moreover, for terms arising from the martingale part,
\begin{equation}
\label{eq:term2}
 \mathbb{E}\left|\sup_{0\leq t\leq T}\int_0^t \nabla_{x_j}\phi_k(X^\epsilon_s,X^\epsilon_s/\epsilon, \ldots,X^\epsilon_s/\epsilon^N)\,dW_s\right|^2 < \infty,
\end{equation}
together with 
\begin{equation}
\label{eq:term3}
\sup_{0\leq t \leq T}|\phi_j(X^\epsilon_t)| < \infty.
\end{equation}
Terms of the type (\ref{eq:term1}) can be bounded above by:
\begin{align*}
\mathbb{E} \sup_{0 \leq t \leq T}\int_0^t \left|\left(\nabla_{x_i}V\cdot\nabla_{x_j} \phi_k\right)(X_s^\epsilon,\ldots, X_s^\epsilon/\epsilon^N)\right| + \left|\sigma\nabla_{x_i}\cdot\nabla_{x_j} \phi_k(X_s^\epsilon,\ldots, X_s^\epsilon/\epsilon^N)\right|\,ds.
\end{align*}
If $i > 0$, then $\nabla_{x_i}V$ is uniformly bounded, and so the above expectation is bounded above by
\begin{align*}
&C\,\mathbb{E}  \int_0^T |\nabla_{x_j}\phi_k(X_s^\epsilon, \ldots, X_s^\epsilon/\epsilon^N)| + |\nabla_{x_i}\cdot\nabla_{x_j}\phi_k(X_s^\epsilon, \ldots, X_s^\epsilon/\epsilon^N)| \,ds \\
=&C \mathbb{E} \int_0^T\sum_{m=1}^3\left\lVert \nabla^{m}_{x_0}\phi_0(X_s^\epsilon)\right\rVert_F\,ds \leq KT,
\end{align*}
using (\ref{eq:phi0bounds}), for some constant $K > 0$.    For the case when $i = 0$, an additional term arises from the derivative $\nabla_{x_0}V_0$ and we obtain an upper bound of the form
\begin{equation}
\label{eq:tight_bound1}
\begin{aligned}
\mathbb{E} \int_0^T&\sum_{m=1}^3\left\lVert \nabla^{m}_{x_0}\phi_0(X_t^\epsilon)\right\rVert(1 + \left|\nabla_{x_0} V_0(X_t^\epsilon)\right|)\,dt \\
&\leq  \,\mathbb{E} \int_0^T\sum_{m=1}^3\left\lVert \nabla^{m}_{x_0}\phi_0(X_t^\epsilon)\right\rVert(1 + \lVert\nabla\nabla V_0\rVert_{L^\infty}|X_t^\epsilon|)\,dt
\end{aligned}
\end{equation}
and which is bounded by Assumption \ref{ass:potential1} and (\ref{eq:phi0bounds}).  For (\ref{eq:term2}), we have
\begin{align*}
\mathbb{E}\left|\sup_{0\leq t\leq T}\int_0^t \nabla_{x_j}\phi_k(X^\epsilon_s,X^\epsilon_s/\epsilon, \ldots,X^\epsilon_s/\epsilon^N)\,dW_s\right|^2 &\leq 4\mathbb{E}\int_0^T |\nabla_{x_j}\phi_k(X^\epsilon_s,X^\epsilon_s/\epsilon, \ldots,X^\epsilon_s/\epsilon^N)|^2\,ds\\
&\leq C\,\mathbb{E}\int_0^T\sum_{m=1}^3\left\lVert \nabla^{m}_{x_0}\phi_0(X_s^\epsilon)\right\rVert_F\,ds,
\end{align*}
which is again bounded.  Terms of the type (\ref{eq:term3}) follow in a similar manner.  Condition \ref{it:tight1} then follows by an application of  Markov's inequality.
\\\\
To prove Condition \ref{it:tight2}, we set $\phi_0(x)=x$ and let $\phi_{1}, \ldots, \phi_{2N-1}$ be the test functions which exist by Proposition \ref{prop:test_functions}.  Applying It\^{o}'s formula to the corresponding multiscale test function (\ref{eq:tightness_test_function}), so that for $t_0 \in [0,T]$ fixed,
\begin{equation}
\label{eq:tightness2_decomposition}
  X_t^{\epsilon} - X_{t_0}^\epsilon = \int_{t_0}^t {G}\,ds + \sum_{i=0}^{N}\sum_{k=0}^{2N-1}\epsilon^i \int_{t_0}^t \nabla_{x_i}\phi_j\,dW_s,
\end{equation}
where $G$ is of the form given in (\ref{eq:tightness_decomposition}). Let $M > 0$, and let 
\begin{equation}
\label{eq:stopping_time}
\tau^\epsilon_M = \inf\lbrace t \geq 0\, ; \, |X_t^\epsilon| > M\rbrace.
\end{equation}
Following \cite{pardoux2001poisson}, it is sufficient to show that
\begin{equation}
\label{eq:tight2_condition1}
\mathbb{E} \left[\sup_{t_0\leq t \leq T} \int_{t_0\wedge \tau^\epsilon_M}^{t\wedge \tau^\epsilon_M}\left|e^{V/\sigma}\nabla_{x_i}\cdot\left(e^{- V/\sigma} \nabla_{j}\phi_{k}\right)(X_s^\epsilon, X_s^\epsilon/\epsilon, \ldots,  X_s^\epsilon/\epsilon^N)\,ds\right|^{1+\nu}\right] < \infty
\end{equation}
and
\begin{equation}
\label{eq:tight2_condition2}
\mathbb{E}\left(\sup_{t_0\leq t \leq t_0 + \gamma}\left|\int_{t_0\wedge \tau^\epsilon_M}^{t\wedge\tau^\epsilon_M} \nabla_{x_i}\phi_j(X_s^\epsilon, X_s^\epsilon/\epsilon, \ldots,  X_s^\epsilon/\epsilon^N)\,dW_s\right|^{2+2\nu}\right) < \infty
\end{equation}
for some fixed $\nu > 0$.   For (\ref{eq:tight2_condition1}), when $i > 0$, the term $\nabla_{x_i}V$ is uniformly bounded.  Moreover, since $\nabla\phi_0$ is bounded, so are the test functions $\phi_1,\ldots, \phi_{2N+1}$.  Therefore, by Jensen's inequality one obtains a bound of the form
\begin{align*}
&C\gamma^{\nu}\mathbb{E}\int_{t_0}^{t_0+\gamma}\left|e^{ V/\sigma}\nabla_{x_i}\cdot\left(e^{- V/\sigma} \nabla_{j}\phi_{k}\right)(X_s^\epsilon, X_s^\epsilon/\epsilon, \ldots,  X_s^\epsilon/\epsilon^N)\right|^{1+\nu}\,ds\\
\leq&C\gamma^{\nu}\int_{t_0}^{t_0+\gamma} |K|^{1+\nu}\,ds \leq K'\gamma^{1+\nu}.
\end{align*}
When $i = 0$, we must control terms involving $\nabla_{x_0}V_0$ of the form, 
$$
\mathbb{E}\left[\sup_{t_0\leq t\leq t_0+\gamma}\int_{t_0\wedge\tau^\epsilon_M}^{t\wedge\tau^\epsilon_M} \left|\nabla V_0\cdot \nabla_{x_j} \phi_k\right|^{1+\nu}\,ds\right]
$$
where $\tau_M^\epsilon$ is given by (\ref{eq:stopping_time}).  However, applying Jensen's inequality,
\begin{align}
\notag\mathbb{E}\left[\sup_{t_0\leq t\leq t_0+\gamma}\int_{t_0\wedge\tau^\epsilon_M}^{t\wedge\tau^\epsilon_M} \left|\nabla V_0\cdot \nabla_{x_j} \phi_k\right|^{1+\nu}\,ds\right] &\leq C\gamma^{\nu}\int_{t_0\wedge \tau^\epsilon_M}^{(t_0+\gamma)\wedge \tau^\epsilon_M} \mathbb{E}\left|\nabla V_0\cdot \nabla_{x_j} \phi_k\right|^{1+\nu}\,ds \\
&\notag\leq C\gamma^{\nu}\int_{t_0\wedge \tau^\epsilon_M}^{(t_0+\gamma)\wedge \tau^\epsilon_M} \mathbb{E}\left|\nabla V_0(X_s^\epsilon)\right|^{1+\nu}\,ds\\
&\notag\leq C\gamma^{\nu}\left\lVert \nabla^2 V_0\right\rVert_{\infty}^{1+\nu}\int_{t_0\wedge \tau^\epsilon_M}^{(t_0+\gamma)\wedge \tau^\epsilon_M} \mathbb{E}|X_s^\epsilon|^{1+\nu}\,ds\\
&\label{eq:tight_bound2}\leq CM\gamma^{1+\nu}\left\lVert \nabla^2 V_0\right\rVert_{L^{\infty}}^{1+\nu},
\end{align}
as required.  Similarly, to establish (\ref{eq:tight2_condition2}) we follow a similar argument, first using the Burkholder-Gundy-Davis inequality to obtain:
\begin{align*}
  \mathbb{E}\left(\sup_{t_0 \leq t \leq t_{0} + \gamma}\int_{t_0}^{t}|\nabla_{x_i}\phi_j\,dW_s|^{2+2\nu}\right) &\leq   \mathbb{E}\left(\int_{t_0}^{t_0+\gamma}\left|\nabla_{x_i}\phi_j\right|^2 \,ds\right)^{1+\nu}\\
  &\leq   \gamma^{\nu}\int_{t_0}^{t_0+\gamma}\mathbb{E}\left|\nabla_{x_i}\phi_j\right|^{2+2\gamma} \,ds \\
  &\leq C \gamma^{1+\nu}.
\end{align*}  
We note that Assumption \ref{ass:potential1} (3) is only used to obtain the bounds (\ref{eq:tight_bound1}) and (\ref{eq:tight_bound2}).  A straightforward application of Markov's inequality then completes the proof of condition \ref{it:tight2}.  It follows from Prokhorov's theorem that the family $\lbrace X_t^\epsilon ; t \in [0,T]\rbrace_{0 < \epsilon \leq 1}$  is relatively compact in the topology of weak convergence of stochastic processes taking paths in $C([0,T]; \mathbb{R}^d)$. In particular, there exists a process $X^0$ whose paths lie in  $C([0,T]; \mathbb{R}^d)$ such that $\lbrace X^{\epsilon_n}; t\in[0,T] \rbrace \Rightarrow \lbrace X^{0}; t\in[0,T] \rbrace $ along a subsequence $\epsilon_n$.

\subsection{Identifying the Weak Limit}
In this section we uniquely identify any limit point the set $\lbrace X_t^\epsilon; t \in [0,T]\rbrace_{0 < \epsilon \leq 1}$.  Given $\phi_0 \in C^\infty_c(\mathbb{R}^d)$ define $\phi^\epsilon$ to be
$$
  \phi^\epsilon(x) = \phi_0(x) + \epsilon \phi_1(x/\epsilon) + \ldots \epsilon^N \phi_N(x, x/\epsilon,\ldots, x/\epsilon^N) + \ldots + \epsilon^{2N}\phi_{2N}(x, x/\epsilon, \ldots, x/\epsilon^N),
$$
where $\phi_1, \ldots, \phi_N$ are the test functions obtained from Proposition \ref{prop:test_functions}.  Since each test function is smooth, we can apply It\^{o}'s formula to  $\phi^\epsilon(X_t^\epsilon)$ to see that
$$
  \mathbb{E}\left[\phi_0(X_t^\epsilon) - \int_{s}^t \frac{\sigma}{Z(X_u^\epsilon)}\nabla_{x_0}\cdot\left({Z}(X_u^\epsilon)\mathcal{M}(X_u^\epsilon)\nabla \phi_0(X_u^\epsilon)\right)\,du + \epsilon R_\epsilon\,\Big|\, \mathcal{F}_s\right] = \phi_0(X_s^\epsilon),
$$
where $R_\epsilon$ is a remainder term which is bounded in $L^2(\pi^\epsilon)$ uniformly with respect to $\epsilon$, and where the homogenized diffusion tensor $\mathcal{M}(x)$ is defined in Theorem \ref{thm:homog_main}.   Taking $\epsilon \rightarrow 0$ we see that any limit point is a solution of the martingale problem
$$
\mathbb{E}\left[\phi_0(X^0_t) - \int_{s}^t \frac{\sigma}{Z(X^0_u)}\nabla_{x_0}\cdot\left(Z(X^0_u)\mathcal{M}(X^0_u)\nabla \phi_0(X_u^0)\right)\,du \,\Big|\, \mathcal{F}_s\right] = \phi_0(X_s^0).
$$
This implies that $X^0$ is a solution to the martingale problem for $\mathcal{L}^0$ given by
$$
  \mathcal{L}_0 f(x) = \frac{\sigma}{Z(x)}\nabla\cdot(Z(x)\mathcal{M}(x)\nabla f(x)).
$$

From Lemma \ref{lem:corr_existence_ellipticity}, the matrix $\mathcal{M}(x)$ is smooth, strictly positive definite and has bounded derivatives.  Moreover,
\begin{align*}
  Z(x) &= \int_{\mathbb{T}^d}\cdots\int_{\mathbb{T}^d}e^{-V(x, x_1,\ldots, x_N)/\sigma}\,dx_1\ldots dx_N \\ &= e^{-V_0(x)/\sigma}\int_{\mathbb{T}^d}\cdots\int_{\mathbb{T}^d}e^{-V_1(x, x_1,\ldots, x_N)/\sigma}\,dx_1\ldots dx_N,
\end{align*}
where the term in the integral is uniformly bounded.  It follows from Assumption \ref{ass:potential1}, that for some $C > 0$, 
$$
  \left|\mathcal{M}(x)\nabla \Psi(x)\right| \leq C(1 + |x|), \quad \forall x\in\mathbb{R}^d,
$$
where $\Psi = -\log Z$.  Therefore, the conditions of the Stroock-Varadhan theorem \cite[Theorem 24.1]{rogers2000diffusions} holds, and therefore the martingale problem for $\mathcal{L}^0$ possesses a unique  solution.  Thus $X^0$ is the unique (in the weak sense) limit point of the family $\lbrace X^\epsilon \rbrace_{0 < \epsilon \leq 1}$.  Moreover, by \cite[Theorem 20.1]{rogers2000diffusions}, the process $\lbrace X^0_t; t\in[0,T] \rbrace$ will be the unique solution of the SDE (\ref{eq:homogenized_sde}), completing the proof.

\section{Acknowledgements}
\label{sec:ack}
The authors thank S. Kalliadasis and M. Pradas for useful discussions. They also thank B. Zegarlinski for useful discussions and for pointing out Ref. \cite{hebisch2010coercive}. We acknowledge financial support by the Engineering and Physical Sciences Research Council of the UK through Grants Nos.  EP/J009636, EP/L020564, EP/L024926 and EP/L025159.

\section*{References}
\bibliographystyle{standard}
\input{refs.bbl}

% \bibliography{refs}
\end{document}

%% file: refs.bbl
\def\cprime{$'$}